\documentclass[12pt]{article}

\usepackage[colorlinks, linkcolor=blue,citecolor=blue]{hyperref}           
\usepackage{color}

\usepackage{graphicx,subfigure,amsmath,amssymb,amsfonts,bm,epsfig,epsf,url,dsfont}
\usepackage{amsthm,mathrsfs}
\usepackage{tikz}
\usepackage{multirow}
\usepackage{bbm}      
\usepackage{booktabs}
\usepackage{overpic}
\usepackage{rotating}
\usepackage{cases}
\usepackage{fullpage,enumitem}
\usepackage[small,bf]{caption}
\usepackage[top=1in,bottom=1in,left=1in,right=1in]{geometry}
\usepackage{fancybox}
\usepackage{hyperref}
\usepackage{algorithm}
\usepackage{verbatim}
\usepackage{algorithmic}
\usepackage{amsthm}
\usepackage{mathtools,subfigure}

\usepackage{scalerel,stackengine}
\stackMath
\newcommand\reallywidehat[1]{%
\savestack{\tmpbox}{\stretchto{%
0  \scaleto{%
    \scalerel*[\widthof{\ensuremath{#1}}]{\kern-.6pt\bigwedge\kern-.6pt}%
    {\rule[-\textheight/2]{1ex}{\textheight}}
  }{\textheight}%
}{0.5ex}}%
\stackon[1pt]{#1}{\tmpbox}%
}
\newenvironment{nscenter}
 {\parskip=0pt\par\nopagebreak\centering}
 {\par\noindent\ignorespacesafterend}

\newcommand{\floor}[1]{\left \lfloor #1 \right \rfloor}

\newcommand{\bE}{\mathbb{E}}

\newtheorem{definition}{Definition}
\newtheorem{example}{Example}

\newtheorem{theorem}{Theorem}

\newtheorem{lemma}{Lemma}
\newtheorem{prop}{Proposition}

\newenvironment{fminipage}%
  {\begin{Sbox}\begin{minipage}}%
  {\end{minipage}\end{Sbox}\fbox{\TheSbox}}

\makeatletter
\newcommand*{\rom}[1]{\expandafter\@slowromancap\romannumeral #1@}
\makeatother
\newtheorem{remark}{Remark}

\newcommand{\Ind}{\mathbbm{1}}

\newcommand{\adv}{\mathsf{{Adv}}}

\newcommand{\sd}{\mu}

\newcommand{\yb}{\mathbf{y}}
\newcommand{\Ib}{\mathbf{I}}
\newcommand{\Jb}{\mathbf{J}}

\newcommand{\abs}[1]{\left|#1\right|}
\newcommand{\R}{\mathbb{R}} 
\newcommand{\N}{\mathbb{N}}

\newcommand{\E}{\mathbb{E}}

\newcommand{\cE}{\mathcal{E}}  
\def\P{{\mathbb P}}

\newcommand {\pr} {\mathbb{P}}

\newcommand{\calA}{{\cal A}}

\newcommand{\calC}{{\cal C}}

\newcommand{\calE}{{\cal E}}

\newcommand{\calG}{{\cal G}}
\newcommand{\calH}{{\cal H}}

\newcommand{\calK}{{\cal K}}

\newcommand{\calN}{{\cal N}}
\newcommand{\calO}{{\cal O}}

\newcommand{\calS}{{\cal S}}
\newcommand{\calT}{{\cal T}}

\newcommand{\calZ}{{\cal Z}}

\DeclarePairedDelimiterX{\set}[1]{\{}{\}}{\setargs{#1}}
\DeclarePairedDelimiterX{\cond}[1]{[}{]}{\setargs{#1}}
\NewDocumentCommand{\setargs}{>{\SplitArgument{1}{;}}m}
{\setargsaux#1}
\NewDocumentCommand{\setargsaux}{mm}
{\IfNoValueTF{#2}{#1} {#1\,\delimsize|\,\mathopen{}#2}}

\newcommand{\be}{\begin{equation}}
\newcommand{\ee}{\end{equation}}
\newcommand{\beqna}{\begin{eqnarray}}
\newcommand{\eeqna}{\end{eqnarray}}

\newcommand{\p}[1]{\left(#1\right)}
\newcommand{\pp}[1]{\left[#1\right]}
\newcommand{\ppp}[1]{\left\{#1\right\}}
\newcommand{\norm}[1]{\left\|#1\right\|}
\newcommand{\innerP}[1]{\left\langle#1\right\rangle}

\usepackage{xspace}
\usepackage{array}
\usepackage{tabularx}
\usepackage{booktabs} 
\usepackage{adjustbox} 

\newcommand{\s}[1]{\mathsf{#1}}

\makeatletter
\def\thanks#1{\protected@xdef\@thanks{\@thanks
        \protect\footnotetext{#1}}}
\makeatother

\makeatletter
\renewcommand{\paragraph}{%
  \@startsection{paragraph}{4}%
  {\z@}{1.25ex \@plus 1ex \@minus .2ex}{-1em}%
  {\normalfont\normalsize\bfseries}%
}
\makeatother

\allowdisplaybreaks

\begin{document}
\title{Robust Detection of Planted Subgraphs in Semi-Random Models}
\author{Dor Elimelech~~~~~~~~~~~~~Wasim Huleihel\thanks{D. Elimelech and W. Huleihel are with the School of Electrical Engineering and Computer Engineering, at Tel Aviv University, {T}el {A}viv 6997801, Israel (e-mails:  \texttt{dorelimelech@tauex.tau.ac.il, wasimh@tauex.tau.ac.il}). This work is supported by the ISRAEL SCIENCE FOUNDATION (grant No. 1734/21).}}

\maketitle
\begin{abstract}
Detection of planted subgraphs in Erd\H{o}s-R\'{e}nyi random graphs has been extensively studied, leading to a rich body of results characterizing both statistical and computational thresholds. However, most prior work assumes a purely random generative model, making the resulting algorithms potentially fragile in the face of real-world perturbations. In this work, we initiate the study of semi-random models for the planted subgraph detection problem, wherein an adversary is allowed to remove edges outside the planted subgraph before the graph is revealed to the statistician. Crucially, the statistician remains unaware of which edges have been removed, introducing fundamental challenges to the inference task.

We establish fundamental statistical limits for detection under this semi-random model, revealing a sharp dichotomy. Specifically, for planted subgraphs with strongly sub-logarithmic maximum density detection becomes information-theoretically impossible in the presence of an adversary—despite being possible for some planted subgraphs in the classical random model. In stark contrast, for subgraphs with super-logarithmic density, the statistical limits remain essentially unchanged; we prove that the optimal (albeit computationally intractable) likelihood ratio test remains robust. Beyond these statistical boundaries, we design a new computationally efficient and robust detection algorithm, and provide rigorous statistical guarantees for its performance. Our results establish the first robust framework for planted subgraph detection and open new directions in the study of semi-random models, computational-statistical trade-offs, and robustness in graph inference problems.
\end{abstract}

\section{Introduction}\label{sec:intro}
Detecting planted structures in large random graphs is a fundamental problem in statistics, computer science, and network analysis. While much of the classical literature has focused on partitioning nodes into communities, an equally fundamental challenge is to determine whether a small, structured subgraph has been embedded in a noisy background—typically modeled as an Erd\H{o}s–R\'{e}nyi random graph. This setting has been extensively studied for specific planted subgraphs such as cliques, dense subgraphs, paths, and trees, leading to a broad understanding of both statistical and computational limits, see, e.g., \cite{butucea2013detection,arias2014community,verzelen2015community,Hajek2015,arias2015detecting,hajek2015computational,chen2016statistical,hajek2016information,brennan18a,brennan19,massoulie19a,Bagaria20,10.1214/20-AAP1660,HuleihelBip}, just to name a few.

Recently, significant progress has been made in characterizing detection thresholds for arbitrary planted subgraphs under a random generative model \cite{elimelech2025detecting}. The setup is defined as follows. Let $n \in \mathbb{N}$, and let $q \in (0,1)$, $p \in (0,1)$, and $\Gamma = \Gamma_n$ be a sequence of arbitrary undirected graphs, referred to as the \emph{planted subgraph}. We consider the following graph detection problem. Under the null hypothesis, the observed graph $\s{G}$ is sampled from the Erd\H{o}s--R\'{e}nyi random graph model $\calG(n,q)$, where each edge is independently included with probability $q$. Under the alternative hypothesis, a uniformly random copy of the subgraph $\Gamma_n$ is selected and embedded into the complete graph on $n$ vertices. Each edge of $\Gamma_n$ is then added to $\s{G}$ independently with probability $p$, while all other edges (those not in $\Gamma_n$) are included independently with probability $q$. In this setting, detection feasibility is governed by properties like the maximum subgraph density and maximal degree, and computational barriers can arise when the planted structure becomes information-theoretically detectable but eludes efficient algorithms. 

Such results establish a rich theory of statistical-to-computational tradeoffs, but they rely heavily on the assumption that the observed graph is drawn from a \emph{purely random model}. In many real-world applications, however, data deviate from such idealized assumptions—whether due to noise, missing data, or even adversarial perturbations. This raises an important question: 

\vspace{0.2cm}
\centerline{\noindent\fbox{\parbox{0.9\textwidth}{
\begin{nscenter}
     \emph{Are existing detection methods robust to adversarial deviations from randomness?}

\end{nscenter}
}}}
\vspace{0.2cm}

\noindent We study this question in the context of \emph{semi-random planted subgraph detection}, where an adversary is allowed to remove edges outside the planted subgraph after its insertion but before the graph is revealed. This semi-random model captures a more realistic and adversarial setting: the planted signal is preserved, but the background is perturbed in unknown ways, potentially obscuring statistical cues used by standard algorithms. The statistician is unaware of which edges have been deleted, making the problem significantly more challenging than in the classical random case.

To formalize this setting, consider again the sequence $\Gamma = (\Gamma_n)_n$ of graphs with $|v(\Gamma_n)| \le n$. We now observe a graph $\cal{G}_{\s{adv}}$ generated under one of two hypotheses:
\begin{itemize}
    \item Under $\calH_0$: a graph $\s{G} \sim \calG(n,q)$ is drawn from the Erd\H{o}s–R\'{e}nyi model. An adversary is allowed to remove arbitrary edges, producing the final observed graph $\s{G}_{\s{adv}_0}$.
    \item Under $\calH_1$: a uniform copy of $\Gamma$ is planted by retaining each of its edges independently with probability $p$, while the remaining edges are included with probability $q$. An adversary then deletes arbitrary edges from outside the planted subgraph, yielding $\s{G}_{\s{adv}_1}$.
\end{itemize}
The goal is to distinguish between $\mathcal{H}_0$ and $\mathcal{H}_1$, given only the final graph $\s{G}_{\s{adv}}$. This adversarial model, also known as the monotone adversary or semi-random deletion model, was originally introduced by \cite{feige2000finding}, and further developed by \cite{feige2001heuristics,mckenzie2020new,buhai2023algorithms,blasiok2024semirandom} in the context of hidden clique problems. It generalizes earlier frameworks in robust inference and semi-random community detection, such as those explored in the planted clique and planted dense subgraph problems. In particular, recent work has established recovery conjectures for the planted dense subgraph (PDS) problem under semi-random perturbations~\cite{brennan20a}, providing strong average-case reduction-based evidence that recovery becomes computationally hard well below the information-theoretic threshold—even when the adversary is restricted to edge deletions.

\paragraph{Main contributions.} In this work, we initiate a systematic study of semi-random planted subgraph detection for general subgraphs and establish both statistical and algorithmic barriers in this setting:
\begin{itemize}
    \item \emph{Statistical limits.} We prove a sharp threshold for detectability under semi-random edge deletions: if the planted subgraph has sub-logarithmic maximum density (see, Definition~\ref{def:maxsubden}), then detection is information-theoretically impossible. This impossibility arises despite the same subgraphs being detectable in the purely random setting. On the other hand, when the subgraph has super-logarithmic density, a relaxation of the likelihood ratio test—although computationally infeasible—remains robust and achieves statistically optimal detection.
    \item \emph{Robust efficient algorithm.} We propose a computationally efficient and robust testing procedure for the detection of planted subgraphs in the presence of monotone adversaries—a setting where existing efficient methods, such as total edge count and maximum degree tests, provably fail. By formulating a convex relaxation of the intractable maximum likelihood estimator (MLE) via a nuclear norm-constrained optimization problem, we derive a semidefinite program (SDP)-based test that remains reliable under the adversarial action. We prove that our test succeeds under a natural and interpretable condition on the planted subgraph, namely, when the ratio of the number of edges to the nuclear norm of the subgraph is at least on the order of $\sqrt{n}$. Notably, our algorithm is, in fact, \textit{computationally optimal} in many cases of planted subgraphs, including the well-studied folklore examples of cliques and complete bipartite graphs. This result establishes the first computationally efficient detection algorithm that is provably robust to monotone adversaries for a broad class of planted subgraphs.
    \end{itemize}
To the best of our knowledge, our paper provides the first general theory for robust detection of arbitrary planted subgraphs under semi-random perturbations, complementing prior results on recovery and extending robustness analysis beyond specific subgraph families.

\paragraph{Related work.} Our work draws on and extends a large body of research on planted subgraph detection under random models. Foundational work on the planted clique~\cite{jerrum1992large,alon1998finding}, planted dense subgraph~\cite{arias2014community,hajek2015computational}, and stochastic block models~\cite{abbe2017community} revealed key phase transitions and computational limits. More recent work has aimed to understand arbitrary planted structures~\cite{Huleihel2022,pmlr-v195-mossel23a,pmlr-v247-yu24a, lee2025fundamental}, providing tight statistical thresholds and low-degree polynomial barriers in the fully random setting. The idea of using semi-random perturbations to model robustness dates back to Blum and Spencer~\cite{blum1995coloring}, who studied helpful adversaries in learning and graph problems. A more adversarial version was proposed by \cite{feige2000finding} and \cite{feige2001heuristics}, who introduced the semi-random planted clique model. There, edge deletions outside the clique rendered spectral algorithms ineffective, despite preserving statistical detectability. Subsequent research extended these ideas to the stochastic block model, where adversaries may alter edges between or outside communities~\cite{charikar2017learning,moitra2016robustCD,banks2021local,bhaskara2024robustness}. Another recent line of work focused on the problem of recovering planted cliques \cite{steinhardt2017does,mckenzie2020new,buhai2023algorithms,blasiok2024semirandom,guruswami2025semirandom}, bipartite graphs \cite{Kumar22}, and $r$-colorable graphs \cite{louis2025robust}, in the presence of a various (monotone and non-monotone) adversarial models.    While some inference procedures (e.g., semidefinite programs) remain robust, many standard algorithms—such as spectral clustering—fail under semi-random perturbations. These findings emphasized the need for robust algorithm design. In the semi-random planted dense subgraph (PDS) setting, \cite{brennan20a} proposed a compelling recovery conjecture: below a certain density threshold, even weak recovery is computationally hard in the presence of edge deletions. They provided strong evidence using average-case reductions from a strengthened planted clique model with secret leakage (PC$\rho$), and highlighted a fundamental gap between information-theoretic feasibility and efficient algorithms under adversarial noise. Beyond graph models, similar statistical-to-computational gaps have been rigorously studied in average-case complexity via frameworks such as low-degree polynomials~\cite{hopkins2017bayesian,gamarnik2020lowdegree}, sum-of-squares hierarchies~\cite{barak2016nearly}, and statistical query models~\cite{feldman2018complexity}. These tools have proven invaluable in understanding algorithmic limits in semi-random and noisy inference settings. 

\paragraph{Notation.}
In this paper, we adopt the following notational conventions. For an integer $n\in \N$ we denote the set $\ppp{1,\dots,n}$ by $[n]$. We denote the set of all $[n]$-subsets of size $i$ by $\binom{[n]}{i}$. We denote the maximum between two numbers $a,b\in \R$  by $a\vee b$. We will use standard asymptotic notations $O,o,\Omega,\omega$ 
to describe the asymptotic relation between sequences $a_n$ and $b_n$. The notation $a_n\ll b_n$ refers to polynomial smaller then in $n$, i.e., $\limsup_n \log_n a_n< \liminf_n \log_n b_n$. We denote the all-one $m\times n$ matrix by $\mathbf{J}_{m\times n}$, the all-one and all-zero vectors of length $m$ by $\mathbf{1}_m$ and $\mathbf{0}_m$ respectively.    We sometimes omit the subindex from our notation when the dimensions are understood from the context. When $\s{A}$ and $\s{B}$ are square matrices of the same size, we use $\innerP{\s{A},\s{B}}$ to denote the Hilbert-Schmidt inner product, given by $\s{Tr}(\s{A}^\top \s{B})$. We denote the nuclear norm of a matrix $\s{A}$ (also known as the trace norm, or Schatten-$1$ norm) as $\norm{\s{A}}_{\star}$, which is given by the sum of singular values of $\s{A}$. We let $\norm{\s{A}}$ denote the spectral norm of $\s{A}$, and $\norm{\s{A}}_F$ denote its Frobenius norm. For two probability measures $\P_0$ and $\P_1$ on the same probability space we often use the total variation distance, the $\chi^2$-divergence, and the Kullback-Leibler (KL) divergence  defined as $d_{\s{TV}}(\P_0,\P_1)=\frac{1}{2}\int|\mathrm{d}\P_0-\mathrm{d}\P_1|$, $\chi^2(\P_0||\P_1)=\int \frac{(\mathrm{d}\P_0-\mathrm{d}\P_1)^2}{\mathrm{d}\P_1}$ and $d_{\s{KL}}(\mathbb{P}_0||\mathbb{P}_1) = \bE_{\mathbb{P}_0}\log\frac{\mathrm{d}\mathbb{P}_0}{\mathrm{d}\mathbb{P}_1}$, respectfully. For the particular case that $\P_0$ and $\P_1$ are $\s{Bern}(p)$ and  $\s{Bern}(q)$, we denote the $\chi^2$ distance by $\chi^2(p||q)=\frac{(p-q)^2}{q(1-q)}$ and the KL divergence $d_{\s{KL}}(p||q)=p\log\frac{p}{q}+(1-p)\log\frac{1-p}{1-q}$. For a finite $S$, we denote by $\s{Unif}(S)$ the uniform distribution i $S$. Finally, we use the following graph notations: we use $\s{G}=(v(\s{G}),e(\s{G}))$ to denote an undirected graph with vertices $v(\s{G})$ and edges $e(\s{G})$. {For graphs $\s{H}$ and $\s{G}$, we write $\s{H}\subseteq \s{G}$ to mean that $\s{H}$ is an \emph{edge-subgraph} of $\s{G}$, i.e., $e(\s{H})\subseteq e(\s{G})$. In this case, we define $v(\s{H})$ to be the set of vertices incident to at least one edge of $\s{H}$ (so isolated vertices are omitted).} Furthermore, since we mostly consider subgraphs $\s{H}\subseteq \s{G}$  without isolated vertices, we will use $|\s{H}|$ to denote the number of edges in $e(\s{H})$. If $|v(\s{G})|\leq n$, we denote by $\calS_{\s{G}}$ the set of all isomorphic copies of $\s{G}$ in the complete graph on $n$, denoted by $\calK_n$, where an isomorphic copy is a subgraph $\s{H}$ of $\calK_n$ such that there is a bijection $f:v(\s{G})\to v(\s{H})$ that satisfies $(v_1,v_2)\in e(\s{G})$ iff $(f(v_1),f(v_2))\in e(\s{H})$.  For two graphs $\s{G}$ and $\Gamma$, we denote the number of copies of $\Gamma$ in $\s{G}$ by $\calN(\Gamma,\s{G})$. For a graph $\s{G}$, $\s{A}_{\s{G}}$ denotes the $|v(\s{G})|\times |v(\s{G})|$ adjacency matrix of $\s{G}$, for which $\s{A}_{\s{G}}(i,j)$ is the indicator that $\set{i,j}\in e(\s{G})$, for all $i,j\in v(\s{G})$.

The rest of this paper is organized as follows. In Section~\ref{sec:model}, we introduce the problem setup and provide some necessary preliminaries. Section~\ref{sec:mainresults} presents our main results, discussions, and examples. Section~\ref{sec:LB} and Section~\ref{sec:UB} are devoted to the the derivation of our lower and upper bounds, respectively. Finally, in Section~\ref{sec:Out} we conclude our paper, and discuss a few directions for future research.

\section{Problem Setup and Preliminaries}\label{sec:model}

In this section, we formally define the detection problem in semi-random models and introduce some preliminaries. Let $\Gamma=(\Gamma_n)_n$ be a sequence of graphs with $|v(\Gamma_n)|\leq n$. To avoid trivialities, we assume that $\Gamma$ has no isolated vertices. We consider the following detection problem, formulated as a semi-random hypothesis testing problem, with two hypotheses $\calH_0$ and $\calH_1$. 

Under $\calH_0$, we draw an Erd\H{o}s--R\'enyi random graph $\mathcal{G}(n,q)$, denoted by $\s{G}$, obtained by independently including each edge of the complete graph $\calK_n$ with probability $q$. An adversary can then remove edges from $\s{G}$, and we observe the resulting graph, denoted by $\s{G}_{\adv_0}$. 

{Under $\calH_1$, we draw a random graph $\s{G}$ as follows: first draw a uniform copy $\Gamma\sim \s{Unif}(\calS_{\Gamma})$. We keep the edges of $\Gamma$ with probability $p$, and the edges outside $\Gamma$ with probability $q$. We denote the ensemble of such planted graphs by $\calG_{\Gamma_n}(n,p,q)$. The adversary constraint under $\calH_1$ is \emph{edge-based}: it may delete any edge $(i,j)\notin e(\Gamma)$, even if both endpoints $i,j$ lie in $v(\Gamma)$, but it may not modify any edge in $e(\Gamma)$ (i.e., all planted edges are protected). Equivalently, the adversary may delete non-planted edges \emph{within} $v(\Gamma)$, but it cannot delete (or add) the planted edges themselves. We observe the resulting graph, denoted by $\s{G}_{\adv_1}$. Throughout, we assume that $p$ and $q$ are fixed constants in $(0,1)$ with $p>q$ that do not depend on $n$.\footnote{In the vanilla model with no adversaries, the case $p<q$ can be reduced to $p>q$ by taking graph complements. In our semi-random setting, taking complements also swaps the monotonicity direction of the adversary: a deletion-only adversary becomes an addition-only adversary in the complemented graph. The $p<q$ case can therefore be treated analogously by considering the complemented model with a monotone \emph{edge-addition} adversary (with planted edges still protected).}}

Formally, $\s{G}_{\adv_0}=\adv_0(\s{G})$ and $\s{G}_{\adv_1}=\adv_1(\s{G},\Gamma)$ are (possibly random) functions that output a graph in $\{0,1\}^{\binom{n}{2}}$ such that,
\begin{align}
    \P_{\calH_{0}}\pp{\bigcap_{(i,j)\in\binom{[n]}{2}}\s{A}_{\s{G}_{\adv_0}}(i,j)\leq \s{A}_{\s{G}}(i,j) }=1,\label{eq:adversfunct1}
\end{align}
and
\begin{align}
    \P_{\calH_1}\pp{\bigcap_{(i,j)\in\binom{[n]}{2}\setminus e(\Gamma)}\s{A}_{\s{G}_{\adv_1}}(i,j)\leq \s{A}_{\s{G}}(i,j),\bigcap_{(i,j)\in e(\Gamma)}\s{A}_{\s{G}_{\adv_1}}(i,j)= \s{A}_{\s{G}}(i,j)}=1.\label{eq:adversfunct2}
\end{align}
We denote the set of all functions $\adv_0$ and $\adv_1$ satisfying \eqref{eq:adversfunct1}--\eqref{eq:adversfunct2} by $\calA_0$ and $\calA_1$, respectively. In case we consider random functions, $\calA_0$ and $\calA_1$ are sets of conditional distributions (given $\P_{\calH_0}$ and $\P_{\calH_1}$, respectfully). In short, we have the following hypothesis testing problem,
\begin{equation}
\begin{aligned}
&\calH_0: \s{G} \sim \pr_{\calH_0}\quad\s{for}\;\s{some}\quad \pr_{\calH_0}\in\s{Adv}_0\p{\calG(n,q)}\\
&\calH_1: \s{G} \sim \pr_{\calH_1}\quad\s{for}\;\s{some}\quad \pr_{\calH_1}\in\s{Adv}_1\p{\calG_{\Gamma_n}(n,p,q)}.\label{eqn:decMain}
\end{aligned}
\end{equation}
Here, with a slight abuse of notation, $\s{Adv}_0\p{\calG(n,q)}$ denotes the set of distributions induced by an adversary that can remove arbitrary edges in $\calG(n,q)$, and $\s{Adv}_1\p{\calG_{\Gamma_n}(n,p,q)}$ denotes the set of distributions induced by an adversary that can remove edges only outside the planted subgraph. Note that the formulation above highlights the fact that we are dealing here with a generalized hypothesis testing problem. {Also, it should be emphasized that our semi-random model follows the standard information-theoretic convention in which the adversary is \emph{not} assumed to be computationally bounded.}

{This model above is a canonical ``sandwich''/monotone semi-random model and serves as a robustness stress test: even allowing deletions only outside $\Gamma$ can substantially reduce the statistical signal-to-noise ratio by sparsifying the background. Moreover, it provides a clean setting in which one can obtain sharp statistical and computational characterizations for \emph{arbitrary} planted subgraphs. We discuss richer adversary classes (e.g., budget-constrained perturbations) in Section~\ref{sec:Out}.}

As a minor note, we would like to mention here that although the monotone adversary model defined above is restricted to removing edges outside the planted subgraph (see \eqref{eq:adversfunct2}), our results extend (almost trivially) to a slightly more general setting in which the adversary may also \emph{add edges within the planted structure}. The proofs carry over without any modifications to the arguments, and the results remain unchanged. Nevertheless, we opt to focus on \eqref{eq:adversfunct2} to remain consistent with the original classical formulation (also known as the Sandwich Model) introduced in \cite{feige2000finding}.

A test is a binary function $\psi:\ppp{0,1}^{\binom{n}{2}}\to \ppp{0,1}$. The risk of a test is the maximal sum of Type I and Type II error probabilities over all adversarial distributions,
\begin{align}
    \s{R}(\psi)&\triangleq\sup_{\substack{\adv_0\in \calA_0\\ \adv_1\in \calA_1 }} \P_{\calH_0}[\psi(\s{G}_{\adv_0})=1]+\P_{\calH_1}[\psi(\s{G}_{\adv_1})=0]\\
    & = \sup_{\substack{\pr_{\calH_0}\in\s{Adv}_0\\\pr_{\calH_1}\in\s{Adv}_1}} \P_{\calH_0}[\psi(\s{G})=1]+\P_{\calH_1}[\psi(\s{G})=0].\label{def:risk}
\end{align}
The optimal risk is defined to be the minimal risk over all tests,
\begin{align}
    \s{R}^{\star}\triangleq \inf_{\psi:\{0,1\}^{\binom{n}{2}}\to\{0,1\}}\s{R}(\psi).
\end{align}
We say that \textit{strong detection} is possible if $\limsup_{n\to\infty}\s{R}^{\star}=0,$ and impossible otherwise. {Throughout, we assume that $(p,q)$ are known to the learner (to enable explicit threshold calibration), and that the planted shape $\Gamma$ is known when required by a given procedure (or, in some cases, that certain summary characteristics of $\Gamma$ are known). We discuss the impact of unknown $(p,q)$ and/or unknown or misspecified $\Gamma$, as well as possible approaches and future directions, in Section~\ref{sec:Out}.}

Our results will be expressed in terms of the following graph theoretic measures. We let $d_{\max}\left(\Gamma_n\right)$ denote the maximum degree in $\Gamma_n$, and we define the density of $\Gamma_n$ as $\eta(\Gamma_n)\triangleq |e(\Gamma_n)|/|v(\Gamma_n)|$.  Finally, we recall the definition of the \emph{maximum subgraph density}.
\begin{definition}[Maximum subgraph density \cite{bollobas_2001}]\label{def:maxsubden}
Let $\s{G}$ be an undirected graph. The maximum subgraph density of $\s{G}$ is
\begin{align}
\sd(\s{G})\triangleq\max\ppp{\eta(\s{H}):\s{H}\subseteq\s{G},\s{H}\neq\emptyset}.\label{eqn:maxDensity}
\end{align}
\end{definition}
{Throughout this paper, ``sub-/super-logarithmic'' refers to the growth of $\mu(\Gamma)$ as a function of the subgraph size $|v(\Gamma)|$: we say that $\Gamma$ has \emph{sub-logarithmic} maximum density if $\mu(\Gamma)=o(\log|v(\Gamma)|)$, and \emph{super-logarithmic} maximum density if $\mu(\Gamma)=\Omega(\log|v(\Gamma)|)$. We further say that $\Gamma$ has \emph{strongly sub-logarithmic} maximum density if $\mu(\Gamma)=o(\log|v(\Gamma)|/\log\log|v(\Gamma)|)$. This terminology is intrinsic to the planted subgraph, and independent of the ambient size $n$. For example, it can be proved that if $\Gamma$ is a clique on $k$ vertices, then $\mu(\Gamma)=\frac{k-1}{2}$, and hence $\Gamma$ has super-logarithmic maximum density. In contrast, it can be proved that if $\Gamma$ is a tree on $k$ vertices (e.g., a path or a star), then $\mu(\Gamma)=\frac{k-1}{k}$, which is $O(1)$ and thus (strongly) sub-logarithmic.} 

Finally, we sometimes suppress the explicit dependence of quantities on the index $n$ when it is clear from context; for example, we denote the planted-graph sequence by $\Gamma=(\Gamma_n)_n$.

\section{Main Results}\label{sec:mainresults}

In this section, we present our main results. We begin with the statistical limits of our problem. We first prove that, roughly speaking, planted subgraphs with sub-logarithmic density are always statistically impossible to detect under the monotone adversary model in \eqref{eqn:decMain}. We then consider planted subgraphs with super-logarithmic density, where we prove the positive results, showing that the optimal (albeit computationally expensive) algorithm is always robust against the monotone adversary. Finally, we focus on the computational aspects of the problem, where we propose a computationally efficient algorithm, analyze its performance, and provide examples of subgraphs for which the adversary imposes no cost on performance.

\subsection{Statistical limits}\label{sec:StatLim}
In this subsection, we establish the statistical limits, setting aside computational considerations, starting with the information-theoretic lower bounds.
\paragraph{Lower bounds.} We have with the following impossibility result.
\begin{theorem}\label{thm:sublogImp}
Consider the detection problem in \eqref{eqn:decMain}, for any fixed pair of edge probabilities $(p,q)$ independent of $n$. Then:
\begin{enumerate}
    \item For any sequence of planted subgraphs $\Gamma$ such that,
\begin{align}
        \mu(\Gamma)=o\p{\frac{\log|v(\Gamma)|}{\log \log|v(\Gamma)|}} \quad \s{and}\quad |v(\Gamma)|\ll n,\label{eq:condSublog}
    \end{align} 
strong detection is statistically impossible, i.e., $\liminf_{n\to\infty}\s{R}^\star>0$.
\item For graphs with super-logarithmic density, i.e., $\mu(\Gamma)\geq \alpha \log |v(\Gamma)|$ for some $\alpha>0$,\footnote{Note that for any $\Gamma$ with $\mu(\Gamma)= \omega(\log |v(\Gamma)|)$, the bound in \eqref{eqn:LBVanil} simplifies to $\mu(\Gamma) \leq \frac{(1-\varepsilon)\log n}{ \log\p{1+\chi^2(p||q)}}$; for $p=1$, this reduces to $\mu(\Gamma) \leq(1-\varepsilon)\log_{1/q} n$.} if
\begin{align}
    \mu(\Gamma) \leq \frac{(1-\varepsilon)\alpha}{2+\alpha \log\p{1+\chi^2(p||q)}}\cdot \log n,\label{eqn:LBVanil}
\end{align}
for any $\varepsilon>0$, then strong detection is statistically impossible, i.e., $\liminf_{n\to\infty}\s{R}^\star>0$.
\end{enumerate}
\end{theorem}
{The lower bound in Theorem~\ref{thm:sublogImp} is ultimately controlled by the probability that the null graph $\s{G}\sim\calG(n,q)$ already contains a copy of the planted graph $\Gamma$. Understanding when $\calG(n,q)$ contains a given (possibly growing) subgraph is a classical topic in random graph theory, going back to Erd\H{o}s--R\'enyi, and is often governed by the appearance threshold dictated by the densest subgraph of $\Gamma$; see, e.g., the general discussion of the containment problem in \cite{janson2011random}. More recently, the \emph{expectation-threshold} framework of Kahn and Kalai~\cite{kahn2007thresholds} (and its refinements, including the fractional version in \cite{FKNPT21} and the proof of the Kahn--Kalai conjecture in \cite{ParkPham24}) provides a general explanation for why such thresholds are typically captured (up to logarithmic factors) by max-density-type parameters. Theorem~\ref{thm:sublogImp} can be viewed as an inference manifestation of this phenomenon: the monotone adversary reduces detection to (a form of) subgraph containment under the null, where maximum density is the relevant statistic. We note that related connections between inference phase transitions for planted subgraph problems and Kahn--Kalai-type thresholds are explored in \cite{mossel2022second,pmlr-v195-mossel23a,lee2025fundamental}.}

Theorem~\ref{thm:sublogImp} shows that detecting subgraphs under \eqref{eqn:decMain} with strongly sub-logarithmic density is statistically impossible. Let us contrast this result with the vanilla case, where no adversary is present. In that setting, it was shown in \cite{elimelech2025detecting} that detecting subgraphs with sub-logarithmic density is possible, provided that
\begin{align}
    |e(\Gamma_n)|\vee d_{\max}^2(\Gamma_n)\gg n,\label{eqn:condCountMaxDeg}
\end{align}
and impossible otherwise. We note that many common subgraphs—such as stars, paths, regular trees, and others—have sub-logarithmic density. Accordingly, we see that even the (arguably weak) monotone adversary completely eliminates the detectability of such subgraphs, despite their detectability in the absence of adversaries.   

{The proof of \eqref{eq:condSublog} in Theorem~\ref{thm:sublogImp} combines two key ideas. First, we lower bound the optimal risk $\s{R}^\star$ by the probability that the null graph $G\sim \calG(n,q)$ contains a copy of $\Gamma_n$; in other words, when $\Gamma_n$ typically appears inside $\calG(n,q)$, the adversary can erase the surrounding edges and make the two hypotheses statistically hard to distinguish. Second, to control this containment probability for growing graphs $\Gamma_n$, we invoke the second-moment framework developed in \cite{mossel2022second}, which is closely related to the expectation-threshold viewpoint of Kahn and Kalai~\cite{kahn2007thresholds} and its subsequent developments \cite{FKNPT21,ParkPham24}, where max-density statistics play a central role. In particular, \cite[Theorem~1]{mossel2022second} implies that $\pr[\calG(n,q)\;\s{contains}\;\Gamma_n]$ is bounded away from zero when \eqref{eq:condSublog} holds, yielding the stated impossibility.}

{The condition \eqref{eq:condSublog} is a simple sufficient condition ensuring the hypothesis needed in the proof (see Proposition~2). In fact, the proof yields a slightly more general sufficient condition (see \eqref{eq:condCopy}), and \eqref{eq:condSublog} is used only as a convenient simplification. We opted to state the theorem under \eqref{eq:condSublog} to keep the presentation clean. Finally, to illustrate \eqref{eq:condSublog}, we consider the following canonical example.
\begin{example}[Star graph]
Consider the case where $\Gamma = K_{1,k}$ be a star on $k+1$ vertices. Since any subgraph of a star is a (smaller) star, we have
\begin{align}
\mu(K_{1,k})=\frac{k-1}{k}<1,
\end{align}
and in particular $\mu(K_{1,k})=O(1)$ for all $k$. Therefore, whenever $k+1\ll n$, condition~\eqref{eq:condSublog} holds and Theorem~\ref{thm:sublogImp} implies that strong detection is statistically impossible \emph{regardless of how large the star is} (as a function of $n$). This highlights how severely the monotone adversary degrades detectability for tree-like structures: planting a star essentially manifests as increasing the degree of a single vertex (the center), and the adversary can erase this signal by deleting additional incident edges adjacent to that center (that are not part of the planted star), making the resulting graph statistically resemble the null model $\calG(n,q)$.
\newline \indent To complement the information-theoretic statement above (which rules out all tests), we include a simple simulation illustrating the effect on a natural test statistic in the star case: the maximum-degree test (which is optimal in the vanilla model for planted stars \cite{elimelech2025detecting}). In Figure~\ref{fig:star_two} we plot the empirical risk (estimated over Monte Carlo trials) in two settings: the vanilla model (solid line) and the monotone adversarial model (dashed line), where the adversary deletes edges outside the planted star independently with probability $1/3$. In the left panel, we let $n$ vary between $50$ and $2000$, and plant a star of size $k=\frac{\sqrt{n}\log n}{d_{\s{KL}}(p\|q)}$, with edge probabilities $p=1$ and $q=1/2$. In the right panel, we fix $n=1200$ and plant a star of size $k=\alpha n$, where $\alpha$ ranges between $0.01$ and $0.25$, again with $p=1$ and $q=1/2$. Each data point is computed from $1000$ Monte Carlo trials under both hypotheses. 
\newline \indent In the vanilla model, the empirical risk decreases toward zero as $n$ (or $k$) increases. In contrast, under the monotone adversary the risk remains close to the coin-flip level of $1/2$; the adversary can largely suppress the degree signal by deleting non-planted edges incident to the planted center vertex. This illustrate the qualitative dichotomy predicted by Theorem~\ref{thm:sublogImp}. While these simulations do not replace the information-theoretic converse (which rules out all tests in the strongly sub-logarithmic regime), they provide a concrete visualization of how the adversary erases the degree signal for tree-like planted structures.
\end{example}
}

\begin{figure}[t]
\centering
\begin{minipage}{0.49\textwidth}
\centering
\begin{overpic}[scale=0.5]{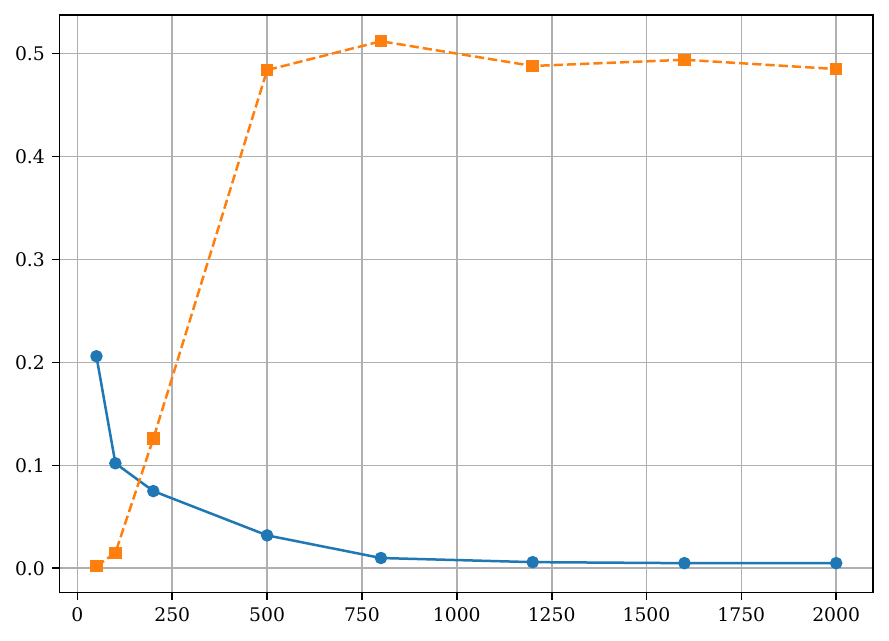}
    \put(50,-4){$n$}
    \put(-5,32){\begin{turn}{90}$\s{Risk}/2$\end{turn}}
\end{overpic}
\label{fig:star_a}
\end{minipage}\hfill
\begin{minipage}{0.49\textwidth}
\centering
\begin{overpic}[scale=0.5]{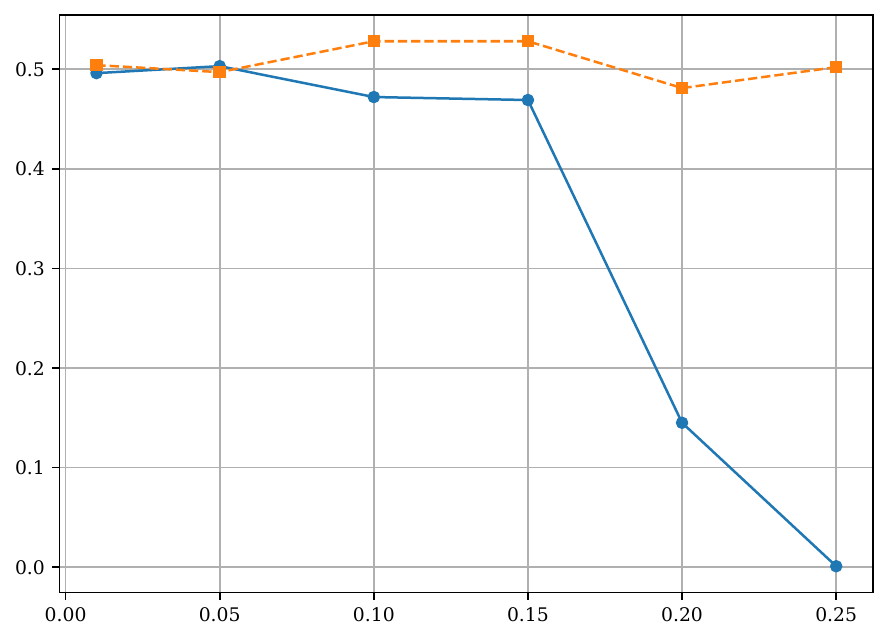}
    \put(50,-4){$\alpha$}
    \put(-5,32){\begin{turn}{90}$\s{Risk}/2$\end{turn}}
\end{overpic}
\label{fig:star_b}
\end{minipage}
\vspace{0.2cm}
\caption{Empirical risk of the maximum-degree test in the vanilla model (solid line) versus under the monotone adversary (dashed line). In the vanilla setting the risk decreases, whereas under the adversary it remains close to $1/2$ (note that we normalize the risk by $1/2$), illustrating the qualitative dichotomy predicted by Theorem~\ref{thm:sublogImp}.}
\label{fig:star_two}
\end{figure}

On the other hand, for graphs with super-logarithmic density, detection is impossible whenever \eqref{eqn:LBVanil} holds. This is precisely the same condition as in \cite[Thm.~12]{elimelech2025detecting}, where impossibility is established for the vanilla (non-adversarial) model. Indeed, by definition of the minimax risk (see \eqref{def:risk}), the same condition immediately implies impossibility for our semi-random testing problem \eqref{eqn:decMain}. Perhaps surprisingly, we next show that \eqref{eqn:LBVanil} is tight by providing an algorithm that is robust to the adversary and succeeds whenever the condition fails.

\paragraph{Upper bounds.} We begin with the observation that the optimal detection algorithm for the case without adversaries is, in fact, robust under the monotone adversary model. In other words, at least statistically, the adversary does not affect the performance of the optimal algorithm. To demonstrate this, we first introduce some notation. Let us now consider the standard planted subgraph detection problem in the absence of any adversary, i.e.,
\begin{align}
\calH_0^{\s{na}}: \s{G} \sim \calG(n,q) \quad \s{vs.} \quad \calH_1^{\s{na}} : \s{G} \sim \calG_{\Gamma_n}(n,p,q).\label{eqn:super_hypo}   
\end{align}
We denote the risk and the optimal risk in the no-adversary model above by $\s{R}_{\s{na}}$ and $\s{R}^\star_{\s{na}}$, respectively. 
{We build on the following ``scan test'' relaxation proposed in~\cite{elimelech2025detecting}.\footnote{Unlike the classical scan test that scans explicitly for $\Gamma$ (see, e.g., \cite{kolar2011minimax,arias2014community,Brennan2018,Huleihel2022}), for arbitrary planted subgraphs we instead scan for subgraphs attaining the maximum subgraph density.}}
Let $\Gamma_{\max}$ be a subgraph that achieves the maximum in the definition of $\mu{\left(\Gamma\right)}$, and then $\calS_{\Gamma_{\max}}$ is the set of all possible copies of $\Gamma_{\max}$ in $\calK_n$. Given an adjacency matrix $\s{A}_n\in\{0,1\}^{n\times n}$, define
\begin{align}
    T_{\s{scan}}(\s{A}_n)&\triangleq \max_{\bar{\Gamma} \in {\calS_{\Gamma_{\max}}}}\sum_{(i,j)\in\bar{\Gamma}}\s{A}_{ij}.\label{eqn:scanstat}
\end{align}
Then, the scan test is defined as $\psi_{\s{scan}}(\s{A}_n)\triangleq\Ind\ppp{T_{\s{scan}}(\s{A}_n)\geq \tau_{\mathsf{scan}}}$, where $\tau_{\mathsf{scan}}\in\mathbb{R}_+$, is specified below. {The statistic $T_{\s{scan}}(\s{A}_n)$ scans over all embeddings $\bar\Gamma$ of $\Gamma_{\max}$ into the ambient vertex set $[n]$ and computes the \emph{overlap} between $\bar\Gamma$ and the observed graph, namely the number of edges of $\s{A}_n$ that fall on the edge set of $\bar\Gamma$. Equivalently, it finds the copy of $\Gamma_{\max}$ whose edge set contains as many observed edges as possible, and then thresholds this maximum overlap. For example, if $\Gamma_{\max}$ is a clique, then $T_{\s{scan}}(\s{A}_n)$ equals the maximum number of edges induced by any $k$-vertex subset; this is the densest-$k$-subgraph objective and is NP-hard to compute in general, illustrating that the scan test is primarily an information-theoretic benchmark rather than a computationally efficient procedure.}

Note that the scan test generally has exponential computational complexity, making it inefficient in practice. Specifically, the search space in \eqref{eqn:scanstat} becomes at least quasi-polynomial when $v(\Gamma_n) = \omega(1)$. 
The following result is proved in \cite[Thm. 8]{elimelech2025detecting}. 
\begin{theorem}[{\cite[Thm. 8]{elimelech2025detecting}}]\label{thm:upperBoundAlgo}
    Consider the detection problem in \eqref{eqn:super_hypo}, and the scan test statistics in \eqref{eqn:scanstat}. Let $\tau_{\s{scan}}= \kappa\cdot |e(\Gamma_{\max})|$, where $\kappa\in(q,p)$. Then, $\s{R}_{{\s{na}}}(\psi_{\s{scan}})\to 0$, provided that $p\cdot |e(\Gamma_{\max})|\to\infty$ and
\begin{align}
    \liminf_{n\to\infty}\frac{\mu(\Gamma)d_{\s{KL}}(p||q)}{\log n}>1.\label{eqn:OptUpper}
\end{align}
\end{theorem}
 As it turns out—almost by construction—the test in \eqref{eqn:scanstat} achieves the same performance stated in Theorem~\ref{thm:upperBoundAlgo}, even when the adversary is present, as in \eqref{eqn:decMain}. This is formalized in the following result.
\begin{theorem}\label{th:advOpt}
Consider the detection problem in \eqref{eqn:decMain}. Then, $\s{R}(\psi_{\s{scan}})\to 0$, under the same conditions of Theorem~\ref{thm:upperBoundAlgo}.
\end{theorem}
Hence, we see that Theorem~\ref{th:advOpt} complements \eqref{eqn:LBVanil} up to a constant factor,\footnote{Notice that for any $\Gamma$ with $\mu(\Gamma)= \omega(\log |v(\Gamma)|)$ and $p=1$, the upper bound in \eqref{eqn:OptUpper} is sharp and exactly matches the lower bound in \eqref{eqn:LBVanil}.} establishing the statistical limits in the super-logarithmic density regime. In other words, putting computational considerations aside, for sufficiently dense subgraphs, the monotone adversary does not impose any statistical cost on the possibility of detection. 
We conclude this subsection with a summary of our results on optimal detection over the semi-random model in \eqref{eqn:decMain}:
\begin{enumerate}
    \item For any $\Gamma$ with $\mu(\Gamma)=o\p{\frac{\log |v(\Gamma)|}{\log \log |v(\Gamma)|}}$, detection is information-theoretically impossible. \emph{This stands in a stark contrast to the completely random model (with no adversaries), where detection is possible (provided that \eqref{eqn:condCountMaxDeg} holds).}
    \item For any $\Gamma$ with $\mu(\Gamma)=\Omega(\log |v(\Gamma)|)$, detection is possible \emph{under the same conditions} as in the completely random model (with no adversaries); i.e., there is no statistical cost for robustness.
    \item For any $\Gamma$ with $\Omega(\frac{\log |v(\Gamma)|}{\log \log |v(\Gamma)|}) \leq \mu(\Gamma)\leq o(\log|v(\Gamma)|)$ a gap remains. We conjecture that in this region the statistical behavior is as in the first item above, that is, detection is information-theoretically impossible. 
\end{enumerate} 
We now proceed to propose robust and efficient algorithms for detection.
\subsection{Computationally efficient and robust detection}
Moving forward to the problem of computationally efficient detection, let us first explain why the current known efficient algorithms for the detection of arbitrary planted subgraphs in the vanilla setting with no adversaries fail in the monotone adversary model. Indeed, as was shown in \cite{elimelech2025detecting}, counting the total number of edges, and evaluating the maximum degree in the observed graph, are the optimal computationally efficient detection algorithms. Specifically, define the statistics
\begin{align}
    T_{\s{count}}(\s{A}_n)&\triangleq \sum_{i<j}\s{A}_{ij},\label{eqn:countstat}\\
    T_{\s{deg}}(\s{A}_n)&\triangleq \max_{i\in[n]}\sum_{j\in[n]} \s{A}_{ij}.\label{eqn:degstat}
\end{align}
We define the corresponding tests as $\psi_{\s{count}}(\s{A}_n)\triangleq\Ind\ppp{T_{\s{count}}(\s{A}_n)\geq \tau_{\mathsf{count}}}$ and $\psi_{\s{deg}}(\s{A}_n)\triangleq\Ind\ppp{T_{\s{deg}}(\s{A}_n)\geq \tau_{\mathsf{deg}}}$, where $\tau_{\s{count}}\triangleq\binom{n}{2}q + |e{\left(\Gamma\right)}| \frac{p-q}{2}$ and $\tau_{\s{deg}}=\left(n-1\right)q + d_{\max}\left(\Gamma\right) \frac{p-q}{2}$. It is shown in \cite[Thm. 8]{elimelech2025detecting} that the tests above succeed provided that \eqref{eqn:condCountMaxDeg} holds. 

Under the monotone adversary model, however, both of these tests fail completely. Indeed, to break both the count and degree tests, under $\calH_1$ let the adversary $\s{Adv}_1$ remove all edges outside the planted subgraph; under $\calH_0$ the adversary keep the graph as is. Then, under $\calH_1$ the observed adjacency matrix has the structure
\begin{align}
\s{A}' =
\begin{bmatrix}
\s{A}_\Gamma & \mathbf{0} \\
\mathbf{0} & \mathbf{0}
\end{bmatrix}.
\end{align}
Accordingly, the number of edges and maximum degree under $\calH_1^{\s{na}}$ is $|e(\Gamma_n)|$ and $d_{\max}(\Gamma_n)$, respectively, while the number of edges and degrees of vertices under $\calH_0^{\s{na}}$ are distributed as $\s{Binomial}(\binom{n}{2},q)$ and $\s{Binomial}(n-1,q)$, respectively. It is clear that the risk of the count and maximum degree tests can be driven to zero only if $|e(\Gamma_n)|=\Omega(\binom{n}{2})$ and $d_{\max}(\Gamma_n)=\Omega(n)$, in a striking contrast to \eqref{eqn:condCountMaxDeg}, when no adversaries exist.

To motivate our robust and efficient algorithm, consider first the problem of \emph{recovering} $\Gamma^\star\in\calS_{\Gamma}$ given an observation $\s{G}\sim\calG_{\Gamma_n}(n,p,q)$. It is not difficult to prove that the maximum-likelihood estimator (MLE) for $\Gamma^\star$ and $p>q$ is given by (see, Appendix~\ref{app:0} for details)
\begin{align}
\hat{\Gamma}_{\s{MLE}} = \arg\max_{\Gamma\in\calS_\Gamma}\sum_{(i,j)\in e(\Gamma)}\s{A}_{ij},\label{eqn:combi}
\end{align}
namely, we aim to maximize the sum of entries among all possible copies of $\Gamma$ in the adjacency matrix $\s{A}$ of the observed graph. Computing the MLE is NP-hard in the worst-case. It is convenient to use the following variation of $\s{A}$, given by
\begin{align}
\s{W}_{ij}\triangleq\begin{cases}
q^{-1}\s{A}_{ij}-1\ &i\neq j\\
0\ &i=j,
\end{cases}\label{eqn:centeredform}
\end{align}
and we can safely replace $\s{A}_{ij}$ with $\s{W}_{ij}$ in \eqref{eqn:combi}. Furthermore, we will represent each $\Gamma\in\calS_\Gamma$ using its corresponding $n\times n$ adjacency matrix, i.e., $[\s{Z}_\Gamma]_{ij}=1$ if and only if $(i,j)\in e(\Gamma)$. Finally, let $\calZ_\Gamma$ denote the set of all such possible matrices, i.e., $\calZ_\Gamma\triangleq\{\s{Z}_\Gamma\in\{0,1\}^{n\times n}:\;\Gamma\in\calS_\Gamma\}$. Accordingly, note that \eqref{eqn:combi} can be equivalently formulated as follows
\begin{align}
\hat{\s{Z}}_{\s{MLE}} = \arg\max_{\s{Z}\in\calZ_\Gamma}\innerP{\s{W},\s{Z}},\label{eqn:combiML}
\end{align}
where $\innerP{\cdot,\cdot}$ is the Hilbert-Schmidt inner product (also known as the Frobenius inner product), defined for two matrices $\s{A}$ and $\s{B}$ as \[\innerP{\s{A},\s{B}} \triangleq \s{Tr}(\s{A}\s{B}^\top)=\sum_{i,j}\s{A}_{i,j}\s{B}_{i,j},\] where $\s{A}$ and $\s{B}$ are of the same dimensions. {Problem \eqref{eqn:combiML} can be viewed as a graph matching problem and is closely related to the quadratic assignment problem (QAP). Indeed, if $|v(\Gamma)|=n$ and $\s{Z}_0$ denotes the adjacency matrix of a \emph{fixed} labeled representative of $\Gamma$ embedded into an $n\times n$ matrix, then maximizing over $\calS_{\Gamma}$ is equivalent to
\begin{align}
    \max_{\s{P}\in\Pi_n}\ \langle \s{W},\s{P} \s{Z}_0 \s{P}^\top\rangle
\end{align}
where $\Pi_n$ is the set of $n\times n$ permutation matrices. When $|v(\Gamma)|=k<n$, the optimization is a subgraph matching version of QAP.}\sloppy

The representation in \eqref{eqn:combiML} typically serves as the starting point for formulating polynomial-time, computationally efficient recovery algorithms using various relaxation techniques. One such technique we build upon is semidefinite programming (SDP), the matrix analogue of linear programming. An SDP can be written in the canonical form:
\begin{equation}
\begin{aligned}
&\underset{\s{Z}\in\mathbb{R}^{n\times n}}{\max}
& & \innerP{\s{C},\s{Z}} \\
& \ \text{s.t.}
& &  \s{Z}\succeq 0\\
&&& \innerP{\s{B}_i,\s{Z}}\leq \beta_i\quad\forall i\in[m],
\end{aligned}\label{eqn:GenSDP}
\end{equation}
for a set of matrices $\s{C},\{\s{B}_i\}_{i=1}^m$, constants $\{\beta_i\}_{i=1}^m$, and $m\in\mathbb{N}$, where $\s{Z}\succeq 0$ indicates that $\s{Z}$ is symmetric positive semidefinite. As convex optimization problems, SDPs are computationally efficient and can be solved using interior-point or first-order methods; see, for example, \cite{nesterov1987interior,boyd2004convex}. A common application of SDP is to approximate solutions to problem with nonconvex constraints, such as integer programs, by using SDP relaxations. 

{In fact, SDP relaxations have a long history of exhibiting robustness in semirandom models. A particularly relevant example is \cite{moitra2016robustCD}, who study community detection under a semirandom model in which the adversary can strengthen within-community connections and weaken across-community connections; they show that semirandom perturbations can shift information-theoretic thresholds, yet an SDP-based algorithm remains effective under such monotone changes. Related robustness guarantees for SDP relaxations in block and synchronization models were established in \cite{FeiChenTIT2020BayesSDP}, who prove that SDP achieves optimal statistical performance while remaining stable under monotone semirandom perturbations. Beyond stochastic block models, semirandom models and SDP-based algorithmic frameworks have been developed for a variety of graph partitioning and inference problems; see, e.g., \cite{MMV12}. More recently, semirandom planted clique models have been studied extensively, with algorithmic guarantees approaching sharp thresholds \cite{BKS23SemirandomCliqueSTOC} and structural connections to restricted isometry properties \cite{BBKS24FOCS}.}

{As will be seen later on, our SDP-based procedure in Theorem~\ref{th:advCON0} fits into this robust-optimization paradigm: while the adversary may alter the observed graph, the monotonicity restriction ensures that the planted structure remains feasible and that the relaxation’s objective continues to favor the planted solution, enabling a proof of provable robustness in our planted subgraph setting.}

To present our proposed testing procedure, let us introduce a few important notations. 
Recall that the nuclear norm $\norm{\s{X}}_{\star}$ is the $\ell_1$-norm of the singular values vector of the matrix $\s{X}$; when $\s{X}$ is symmetric, then $\norm{\s{X}}_{\star}$ is the sum of the absolute values of its eigenvalues. With a slight abuse of notation, for a graph $\s{G}$ we write $\norm{\s{G}}_\star$ for the nuclear norm of the adjacency matrix of $\s{G}$, understood as the subgraph on $|v(\s{G})|$ vertices; note that if $|v(\s{G})|\leq n$, the nuclear norm is unaffected by zero-padding when $\s{G}$ is viewed as a subgraph of an $n$-vertex graph. Fix an arbitrary $\Gamma\in\calS_\Gamma$. Then, for any subgraph $\Gamma'\subseteq \Gamma$, consider the following convex optimization problem
\begin{equation}
\begin{aligned}
\hat{\s{Z}}_{\calC_{\Gamma'}}=& \underset{\s{Z}\in\mathbb{R}^{n\times n}}{\arg\max}
& & \innerP{\s{W},\s{Z}} \\
& \ \text{s.t.}
& &  
\norm{\s{Z}}_\star\leq\norm{\Gamma'}_\star\\
&&& \;\mathbf{0}\leq\s{Z}\leq\mathbf{J},\;\s{Z}=\s{Z^\top},
\end{aligned}\label{eqn:convex}
\end{equation}
where the inequality $\mathbf{0}\leq\s{Z}\leq\mathbf{J}$ is entrywise. {Note that the SDP in~(22) depends on $\Gamma'$ only through the scalar $\norm{\Gamma'}_\star$, and not through any other structural property of $\Gamma'$.} It can be shown that equation~\eqref{eqn:convex} is, in fact, a direct relaxation of the MLE in many special cases (e.g., when $\Gamma$ is a clique, bipartite graph, etc.). Furthermore, it is standard practice to convert the nuclear norm constraint in~\eqref{eqn:convex} into a SDP formulation using the following lemma.
\begin{lemma}[{\cite[Lemma 2]{fazel2002matrix}}]
 Fix $t\geq0$. Let $\s{X} \in \mathbb{R}^{m\times n}$. Then, $\norm{\s{X}}_\star \leq t$ if and only if there exist $\s{W}_1 \in \mathbb{R}^{m \times m}$ and $\s{W}_2 \in \mathbb{R}^{n \times n}$ such that,
    \begin{align}
\begin{bmatrix}
\s{W}_1 & \s{X} \\
\s{X}^\top & \s{W}_2
\end{bmatrix} \succeq 0 \quad \s{and} \quad \s{Tr}(\s{W}_1) + \s{Tr}(\s{W}_2) \leq 2t.\label{eqn:ccooc}
\end{align}
\end{lemma}
Consequently, the optimization problem above can be solved efficiently. Below, we let $\calC_{\Gamma'}(\mathbf{W})$ denote the optimized objective value in~\eqref{eqn:convex}, i.e.,
\begin{align}
\calC_{\Gamma'}(\s{W})\triangleq\innerP{\s{W},\hat{\s{Z}}_{\calC_{\Gamma'}}}.\label{eqn:TestConObject}
\end{align}
Then, we define the proposed test as
\begin{align}
\psi_{\calC_{\Gamma'}}(\s{W})\triangleq\Ind\ppp{\calC_{\Gamma'}(\s{W})>\tau},\label{eqn:TestCon}
\end{align}
with $\tau\in[\tau_0,\tau_1)$, where 
\begin{align}
\tau_0\triangleq \frac{C\norm{\Gamma'}_\star}{q}\pp{\sqrt{n}+\sqrt{\log n}},\quad
\tau_1\triangleq 2\frac{p-q}{q}|e(\Gamma')|-\sqrt{\frac{|e(\Gamma')|}{2}\log n},
\end{align}
for some universal constant $C>0$. 
Note that in \eqref{eqn:convex}, the subgraph $\Gamma'\subseteq \Gamma$ is treated as an input, and may be arbitrarily selected/constructed by the ``user/statistician". While any choice is acceptable, we will see that in some classical cases setting $\Gamma'=\Gamma$ yields the best result. The following result shows that \eqref{eqn:convex} is robust against the monotone adversary in \eqref{eqn:decMain}. 
\begin{theorem}\label{th:advCON0}
Consider the detection problem in \eqref{eqn:decMain}, and fix $\Gamma'\subseteq\Gamma$. Then, $\s{R}(\psi_{\calC_{\Gamma'}})\to0$, provided that 
   \begin{align}
   \frac{|e(\Gamma')|}{\norm{\Gamma'}_\star}\geq C'\sqrt{n},\label{eqn:CondCoPsiGammaPrime}  
   \end{align}
for some constant $C'=C'(p,q)>0$.\footnote{Explicit expression for $C'(p,q)$, as a function of $p$ and $q$, can be found in the proof of Theorem~\ref{th:advCON0}.}
\end{theorem}
Given that Theorem~\ref{th:advCON0} holds for any subgraph $\Gamma'\subseteq\Gamma$, it follows immediately that one may optimize over the choice of $\Gamma'$ to maximize the left-hand side of \eqref{eqn:CondCoPsiGammaPrime}, thereby achieving the strongest possible guarantee. In particular, if we let 
\begin{align}
    \Gamma'_{\s{max}}\in\underset{\Gamma'\subseteq \Gamma}{\arg\max} \frac{|e(\Gamma')|}{\norm{\Gamma'}_\star},\label{eqn:OptiMax}
\end{align}
then applying Theorem~\ref{th:advCON0} with $\Gamma'_{\s{max}}$ yields that $\s{R}(\psi_{\calC_{\Gamma'_{\s{max}}}})\to0$, provided that 
\begin{align}
   \underset{\Gamma'\subseteq \Gamma}{\max}\frac{|e(\Gamma')|}{\norm{\Gamma'}_\star}\geq C'\sqrt{n}.\label{eqn:CondCoPsiGammaPrime2}  
\end{align}
{Importantly, the procedure in \eqref{eqn:convex} is polynomial-time for any \emph{fixed}, a priori choice of $\Gamma'$, and our statistical guarantees hold for any such choice; $\Gamma'$ is an \emph{input choice} to the algorithm made by the statistician. If one instead seeks to \emph{optimize} over $\Gamma'$ according to \eqref{eqn:OptiMax}, then the resulting selection problem is a separate combinatorial optimization task that is not known to be efficiently solvable in general, as it involves combinatorial search over subgraphs with a non-convex objective. However, one can resort to convex relaxations or approximations. For example, it is not difficult to show that $\max_{\s{X}\in\mathbb{R}^{n\times n}}\innerP{\s{X},\Gamma}$ subject to $\norm{\s{X}}_{\star}\leq 1$, $\mathbf{0}\leq\s{X}\leq\mathbf{J}$, and $\s{X}_{ij}=0$ whenever $[\Gamma]_{ij}=0$, is a convex relaxation of \eqref{eqn:OptiMax}. Thresholding the resulting matrix $\s{X}$ then yields an efficiently computable heuristic candidate for $\Gamma'_{\s{max}}$. Further details can be found in Appendix~\ref{app:conRel}. Finally, one may view the best achievable bound (over choices of $\Gamma'$) as being governed by the ``spectral maximum density'' $\max_{\Gamma'} |e(\Gamma')|/\norm{\Gamma'}_\star$. This quantity parallels $\mu(\Gamma)$ with $|v(\cdot)|$ replaced by a nuclear-norm surrogate, and discrepancies between the two explain when statistical and computational thresholds do not coincide.}

While the detection algorithm $\psi_{\calC_{\Gamma'}}$ is general and achieves strong non-trivial guarantees for arbitrary planted subgraph $\Gamma$, its performance does not match that of the simpler count and maximum degree tests in \eqref{eqn:countstat}--\eqref{eqn:degstat} (and in particular \eqref{eqn:condCountMaxDeg}), which are known to be \emph{computationally optimal} in the vanilla random model. This discrepancy/gap is not coincidental: Theorem~\ref{thm:sublogImp} shows that certain subgraphs—such as the star—are undetectable in the semi-random setting. In contrast, the maximum degree test succeeds in the vanilla model as soon as the star's maximum degree (i.e., its number of vertices minus one) exceeds $\sqrt{n\log n}$. Therefore, if $\psi_{\calC_{\Gamma'}}$—or any algorithm—were to match the performance of the maximum degree test in full generality, it would contradict the impossibility result of Theorem~\ref{thm:sublogImp}. This highlights the possibility of an inherent computational barrier in the semi-random model that does not exist in the purely random case. To appreciate Theorem~\ref{th:advCON0}, let us provide a few concrete examples. 
\begin{example}[Clique]
Consider the case where $\Gamma$ is clique, i.e., a complete graph on $|v(\Gamma)|$ vertices. The adjacency matrix of $\Gamma$ in this case has all off-diagonal entries equal to unity and diagonal entries equal to zero. The eigenvalues of the adjacency matrix of $\Gamma$ are well-known:
\begin{align}
    \lambda_1 = |v(\Gamma)|-1, \quad \lambda_2 = \lambda_3 = \cdots = \lambda_{|v(\Gamma)|} = -1.
\end{align}
Therefore, the nuclear norm is given by,
\begin{align}
   \norm{\Gamma}_\star =  2(|v(\Gamma)|-1).
\end{align}
Accordingly, since $|e(\Gamma)| = \binom{|v(\Gamma)|}{2}$ Theorem~\ref{th:advCON0} with $\Gamma'=\Gamma$, implies that $\psi_{\calC_{\Gamma}}$ achieves robust strong detection provided that $|v(\Gamma)|\geq C\sqrt{n}$, for some $C>0$. This coincides with the performance of state-of-the-art detection algorithms in the planted clique problem without adversaries, e.g., \cite{alon1998finding,montanari2015finding}.
\end{example}
\begin{example}[Complete bipartite]\label{example:BiP}
Consider the case where $\Gamma$ is a complete bipartite with partitions of size $k_{\s{L}}$ and $k_{\s{R}}$. The adjacency matrix of $\Gamma$ has the block structure:
\begin{align}
\s{A}_\Gamma =
\begin{bmatrix}
\mathbf{0} & \mathbf{J}_{k_{\s{L}} \times k_{\s{R}}} \\
\mathbf{J}_{k_{\s{R}} \times k_{\s{L}}} & \mathbf{0}
\end{bmatrix}.
\end{align}
The eigenvalues of this matrix are:
\begin{align}
\lambda_1 = \sqrt{k_{\s{L}}k_{\s{R}}}, \quad \lambda_2 = -\sqrt{k_{\s{L}}k_{\s{R}}}, \quad \lambda_3 = \cdots = \lambda_{k_{\s{L}}+k_{\s{R}}} = 0.
\end{align}
Hence, the nuclear norm is:
\begin{align}
\norm{\Gamma}_\star = 2\sqrt{k_{\s{L}}k_{\s{R}}}.
\end{align}
Accordingly, since $|e(\Gamma)| = k_{\s{L}}k_{\s{R}}$ Theorem~\ref{th:advCON0} implies that $\psi_{\calC_{\Gamma}}$ (i.e., we choose $\Gamma'=\Gamma$) achieves robust strong detection provided that $k_{\s{L}}k_{\s{R}}\geq C\cdot n$, for some $C>0$, i.e., $|e(\Gamma)|\geq C\cdot n$. Therefore, $\psi_{\calC_{\Gamma}}$ mimics the performance of the count test in the planted bipartite subgraph detection problem without adversaries.\footnote{In fact, we obtain an even better asymptotic dependence here than the one required by the count test, which needs $|e(\Gamma)|=\Omega(n\sqrt{\log n})$ for strong detection.} {Furthermore, it is consistent with how state-of-the-art recovery algorithms perform on the planted balanced bipartite (or bi-clique) problem, e.g., \cite{Levanzov2018,Kumar2022}.}

\end{example}

\begin{example}[Tur\'{a}n graph]
Consider the case where $\Gamma=T(|v(\Gamma)|,r)$, the complete $r$-partite graph on $|v(\Gamma)|$ vertices that avoids a $(r+1)$-clique. It is the extremal graph that has the maximum number of edges without containing $(r+1)$-clique. As is well-known (see, e.g., \cite{nikiforov2017norms}), we have $\norm{\Gamma}_\star\leq2(1-1/r)|v(\Gamma)|$; in fact, this bound is true for any complete $r$-partite graph of order $n$. Furthermore, $|e(\Gamma)| = (1-1/r)\frac{|v(\Gamma)|^2}{2}$. Thus, Theorem~\ref{th:advCON0} with $\Gamma'=\Gamma$, implies that $\psi_{\calC_{\Gamma}}$ achieves robust strong detection provided that $|v(\Gamma)|\geq C\cdot \sqrt{n}$, for some $C>0$. 
\end{example}

\begin{example}[When a proper subgraph performs better]
    In the above examples, the choice $\Gamma'=\Gamma$ (and thus using $\psi_{\calC_{\Gamma}}$) in Theorem~\ref{th:advCON0} suffices to achieve optimal, state-of-the-art results. However, this need not hold in general: asymptotically improved performance can be guaranteed by using a proper subgraph. 
    
    Consider the following example. Let $\Gamma_1$ be a complete bipartite graph with $k_R=\floor{k^{1/3}}$, $k_L=k$, and let $\Gamma_2$ be a path of length $k$, for an integer $k=k_n$ (that grows with $n$). Define $\Gamma$ to be the disjoint union of $\Gamma_1$ and $\Gamma_2$. It is easy to check that 
    \begin{align}
        \frac{|e(\Gamma)|}{\norm{\Gamma}_\star}=\frac{|e(\Gamma_1)|+|e(\Gamma_2)|}{\norm{\Gamma_1}_\star+\norm{\Gamma_2}_\star}.
    \end{align}
    Furthermore, it is well-known that the eigenvalues associated with the adjacency matrix of a path of length $k$ are given by (see, e.g., \cite[Sec. 1.4.4]{BrouwerHaemers2012})
    \begin{align}
        \lambda_i=2\cos\p{\frac{i\pi}{k+1}}, \quad i=1,\dots,k,
    \end{align}
    and we therefore have
    \begin{align}
        \norm{\Gamma_2}_\star &=\sum_{i=1}^k 2\abs{\cos\p{\frac{i\pi}{k+1}}}=2k\cdot \frac{1}{k}\sum_{i=1}^k \abs{\cos\p{\frac{i\pi}{k+1}}}\\
        &=2k\cdot (1+o(1))\cdot\intop_{0}^{1}|\cos(\pi x)|dx=(1+o(1))\frac{4 }{\pi}k.
    \end{align}
   Now, on the one hand, combining the above with the derivations in Example~\ref{example:BiP}, taking $\Gamma'=\Gamma$ in Theorem~\ref{th:advCON0} yields
    \begin{align}
        \frac{|e(\Gamma')|}{\norm{\Gamma'}_\star}\approx \frac{k^{\frac{4}{3}}+k}{2k^{\frac{2}{3}}+4k/\pi }\approx k^{1/3}.
    \end{align}
    On the other hand, taking $\Gamma'=\Gamma_1$ in Theorem~\ref{th:advCON0} yields
    \begin{align}
        \frac{|e(\Gamma')|}{\norm{\Gamma'}_\star}\approx k^{\frac{2}{3}}\gg k^{1/3}.
    \end{align}
    Thus, Theorem~\ref{th:advCON0} shows that $\psi_{\calC_{\Gamma}}$ achieves strong detection if $k>Cn^{3/2}$, which is clearly impossible as $k\leq n$. In contrast, the same theorem guarantees that $\psi_{\calC_{\Gamma_1}}$ achieves strong detection if $k>Cn^{3/4}$, which is non-trivial and strictly improves upon the previous bound.
\end{example}

We conclude this section by noting that in Appendix~\ref{app:1}, we pay particular attention to the canonical special case where $\Gamma$ is a clique, and provide a complete analysis of yet another SDP relaxation of the MLE in \eqref{eqn:combiML} along with a proof for its robustness against monotone adversaries.

\section{Proof of Lower Bound}\label{sec:LB}

In this section we prove Theorem~\ref{thm:sublogImp}. {The impossibility statement under condition~\eqref{eqn:LBVanil} follows immediately from \cite[Thm.~12]{elimelech2025detecting}, where impossibility is established for the non-adversarial (vanilla) model. Indeed, by definition of the minimax risk in~\eqref{def:risk}, we have
\begin{align}
    \s{R}^{\star} &= \inf_{\psi:\{0,1\}^{\binom{n}{2}}\to\{0,1\}}\sup_{\substack{\pr_{\calH_0}\in\s{Adv}_0\\\pr_{\calH_1}\in\s{Adv}_1}} \P_{\calH_0}[\psi(\s{G})=1]+\P_{\calH_1}[\psi(\s{G})=0]\\
    &\geq \inf_{\psi:\{0,1\}^{\binom{n}{2}}\to\{0,1\}}\P_{\calH_0^{\s{na}}}[\psi(\s{G})=1]+\P_{\calH_1^{\s{na}}}[\psi(\s{G})=0],\label{eqn:optRiskNoadv}
\end{align}
where $\P_{\calH_0^{\s{na}}}$ and $\P_{\calH_1^{\s{na}}}$ denote the distributions under the null and alternative hypotheses in the vanilla detection problem~\eqref{eqn:super_hypo}. The right-hand side of \eqref{eqn:optRiskNoadv} is precisely the optimal risk of the vanilla problem, and therefore any lower bound on the vanilla minimax risk (in particular, \cite[Thm.~12]{elimelech2025detecting}) also lower bounds $\s{R}^\star$ in our semi-random setting. Thus, it remains to prove the lower bound in the sub-logarithmic regime, i.e., \eqref{eq:condSublog}.}

The proof of \eqref{eq:condSublog} follows by a combination of two main observations. The first is stated in the following proposition.
\begin{prop} \label{prop:boundWithProb}
Fix a sequence of subgraphs $\Gamma=\Gamma_n$, and let $\calE$ denote the event that $\calG(n,q)$ contains a copy of $\Gamma_n$. Then, $\s{R}^{\star}\geq \pr\pp{\calE}$, and, in particular, strong detection in the semi-random model is statistically impossible if $\liminf_{n\to\infty}\P[\calE]>0$.
\end{prop}

\begin{proof}[Proof of Proposition~\ref{prop:boundWithProb}]
    Consider the following adversary $\adv_0^\star$:
    \begin{itemize}
        \item If $\s{G}$ contains a copy of $\Gamma$, it chooses a random copy, uniformly over all copies of $\Gamma$ in $\s{G}$, removes all other edges in the graph outside of that copy, and keeps only the edges of the chosen copy with probability $p$ (remove edges with probability $1-p$).
        \item If $\s{G}$ does not contains a copy, it remains unchanged. 
    \end{itemize}
    We also let $\adv_1^\star$ be an adversary that removes any edges outside of $\Gamma$. Then, 
    \begin{align}
        \s{R}^{\star}&= \inf_{\psi:\{0,1\}^{\binom{n}{2}}\to\{0,1\}}\sup_{\substack{\adv_0\in \calA_0 \label{eq:RstarB1}\\ \adv_1\in \calA_1 }} \P_{\calH_0}[\psi(\s{G}_{\adv_0})=1]+\P_{\calH_1}[\psi(\s{G}_{\adv_1})=0]\\
        &\geq \inf_{\psi:\{0,1\}^{\binom{n}{2}}\to\{0,1\}}\P_{\calH_0}[\psi(\s{G}_{\adv^\star_0})=1]+\P_{\calH_1}[\psi(\s{G}_{\adv^\star_1})=0]\\
        & = 1-d_{\s{TV}}\p{\pr_{\calH_0^\star},\pr_{\calH_1^\star}}, \label{eq:RstarB2}
    \end{align}
    where $\pr_{\calH_0^\star}$ and $\pr_{\calH_1^\star}$ denote the probability distributions under the null and alternative distributions after $\adv_0^\star$ and $\adv_1^\star$ were applied, respectively. Recall that $\calE$ denotes the event that $\calG(n,q)$ contains a copy of $\Gamma_n$. Then, we note that,
    \begin{align}
        d_{\s{TV}}\p{\pr_{\calH_0^\star},\pr_{\calH_1^\star}}&= d_{\s{TV}}\p{\P[\cE]\cdot \pr_{ \calH_0^\star\vert\cE }+\P[\cE^c]\cdot \pr_{\calH_0^\star\vert\cE^c},\pr_{\calH_1^\star}} \label{eq:dtv1}\\
        &\leq \P[\cE]\cdot d_{\s{TV}}\p{ \pr_{ \calH_0^\star\vert\cE },\pr_{\calH_1^\star}}+\P[\cE^c]\cdot d_{\s{TV}}\p{  \pr_{\calH_0^\star\vert\cE^c},\pr_{\calH_1^\star}}\\
        & \leq d_{\s{TV}}\p{\pr_{\calH_0^\star\vert\s\cE},\pr_{\calH_1^\star}}+\pr\pp{\calE^c},\label{eq:dtv2}
    \end{align}
    where the first inequality follows from the convexity $(P,Q)\mapsto d_{\s{TV}}(P,Q)$, and the second inequality follows from the facts that $d_{\s{TV}}(P,Q)\leq1$ and $\pr\pp{\calE}\leq 1$. By construction, conditioned on the event $\cE$, we clearly have $d_{\s{TV}}\p{\pr_{\calH_0^\star\vert\cE},\pr_{\calH_1^\star}}=0$, as the two measures can essentially be simulated by planting a copy of $\Gamma$ uniformly at random in an empty graph of $n$ vertices. A proof of this fact is given in Appendix~\ref{app:UniformC}. Combining \eqref{eq:RstarB1}--\eqref{eq:RstarB2} with \eqref{eq:dtv1}--\eqref{eq:dtv2} we obtain that
    \begin{align}
        \s{R}^{\star}&\geq 1-\pr\pp{\calE^c} = \pr\pp{\calE}.
    \end{align}

\end{proof}
\begin{remark}
{Note that the adversary constructed in Proposition~\ref{prop:boundWithProb} requires solving a computationally hard subgraph search problem; this is consistent with our model, which does not impose computational constraints on the adversary.}
\end{remark}
In light of Proposition~\ref{prop:boundWithProb}, our goal now is to lower bound the probability $\pr[\calE]$. We have the following result.
\begin{prop}
\label{prop:sublogProb}
    Let $\Gamma=(\Gamma_n)$ be a sequence of graphs such that $\omega(1) \leq |v(\Gamma)|\leq n $  and
    \begin{align}
        (1+\varepsilon)\mu(\Gamma_n)\cdot \log\log|e(\Gamma_n)|+\log|v(\Gamma_n)|\leq \log n,\label{eq:condCopy}
    \end{align}
    for some $\varepsilon>0$. Then, for any fixed $q$, we have $\pr[\calE]\geq1/2$.
\end{prop}
Before we prove Proposition~\ref{prop:sublogProb} let us show that together with Proposition~\ref{prop:boundWithProb}, it implies Theorem~\ref{thm:sublogImp}.
\begin{proof}[Proof of Theorem~\ref{thm:sublogImp}]
    The proof follows immediately from Propositions~\ref{prop:boundWithProb} and \ref{prop:sublogProb}. Indeed, in the notation of Proposition~\ref{prop:boundWithProb}, we have 
    \begin{align}
        \s{R}^\star\geq \P[\calE]\geq \frac{1}{2},
    \end{align}
    where the second inequality follows by \eqref{eq:condSublog} because,
    \begin{align}
        (1+\varepsilon)\mu(\Gamma_n)\cdot \log\log|e(\Gamma_n)|+\log|v(\Gamma_n)|<\log n,
    \end{align}
    for sufficiently large $n$, in particular, the condition in Proposition~\ref{prop:sublogProb} is satisfied.
\end{proof}

To prove Proposition~\ref{prop:sublogProb}, let us recall the concept of critical probability. For graphs $\Gamma$ and $\s{G}$ let us denote the number of copies of $\Gamma$ in $\s{G}$ by $\calN(\Gamma,\s{G})$. The critical probability of $\Gamma$ is defined as,
\begin{align}
    q_c(\Gamma)&\triangleq\min\ppp{q\in[0,1] ~\bigg|~ \P_{\s{G}\sim \calG(n,q)}\pp{\calN(\Gamma,\s{G})\geq 1}\geq \frac{1}{2}}\\
    &=\min\ppp{q\in[0,1] ~\bigg|~ \P\pp{\calE}\geq \frac{1}{2}}.
\end{align} 
In \cite{mossel2022second}, the critical probability was bounded by a modified subgraph expectation threshold, defined as follows,
\begin{align}
    \Tilde{q}_E(\Gamma)\triangleq\min\ppp{q\in[0,1] ~\bigg|~ \E_{\s{G}\sim \calG(n,q)}\pp{\calN(\s{H},\s{G})}\geq \frac{\calN(\s{H},\Gamma)}{2}\text{ for all }\s{H}\subseteq \Gamma},
\end{align} 
where only subgraphs $\s{H}\subseteq \Gamma$ with no isolated vertices are considered.
\begin{theorem}\cite[Theorem 1]{mossel2022second} \label{th:Zadik} There exists a universal constant $C$ such that for any graph $\Gamma$, \begin{align}\label{eq:critProb}
    \Tilde{q}_E(\Gamma)\leq q_c(\Gamma)\leq C\cdot \Tilde{q}_E\cdot  \log|e(\Gamma)|.
\end{align}
\end{theorem}
Another argument that will be needed for the proof of Proposition~\ref{prop:sublogProb} involves bounding the probability that a uniform random copy of $\Gamma$ in $\calK_n$ contains an arbitrary isomorphic copy of a subgraph $\s{H}\subseteq \Gamma$ in $\calK_n$, namely, $\P_{\Gamma}[\s{H}\subseteq \Gamma]$.  

\begin{lemma}\label{lem:probcalc} For any $\s{H}\subseteq \Gamma$,
\begin{align}
    \P_{\Gamma}[\s{H}\subseteq \Gamma]\leq \p{\frac{|v(\Gamma)|}{n}}^{|v(\s{H})|}.
\end{align}
\end{lemma}

\begin{proof}
    The proof relies on two simple argument. First, we note that by symmetry considerations, the probability that a uniform random of copy $\Gamma$ contains a fixed copy of $\s{H}$ equals to the probability that a uniform random copy of $\s{H}$ in $\calK_n$ is contained in some fixed copy of $\Gamma$ (see, \cite[Lemma~10]{elimelech2025detecting}, for a formal proof). Second, we note that the event that the random copy $\s{H}$ is contained in some fixed copy $\Gamma'$ contains the event that $v(\s{H})\subseteq v(\Gamma')$. Thus, since $v(\s{H})\sim \s{Unif}\binom{[n]}{|v(\s{H})|}$ we have
    \begin{align}
         \P_{\Gamma}[\s{H}\subseteq \Gamma]&\leq \P_{\s{H}}[v(\s{H})\subseteq v(\Gamma)]=\frac{\binom{|v(\Gamma)|}{|v(\s{H})|}}{\binom{n}{|v(\s{H})|}}\\
         &=\frac{|v(\Gamma)|\cdot (|v(\Gamma)|-1)\cdots (|v(\Gamma)|-|v(\s{H})|+1) }{n(n-1)\cdots (n-|v(\s{H})|+1)}\leq \p{\frac{|v(\Gamma)|}{n}}^{|v(\s{H})|},
    \end{align}
    where the last equality follows since the function $f(i)=\frac{k-i}{n-i}$ is monotone decreasing for $k<n$.
\end{proof}
We are now ready to prove Proposition~\ref{prop:sublogProb}.
\begin{proof}[Proof of Proposition~\ref{prop:sublogProb}]
    From the definition of $q_c(\Gamma)$, our goal is to understand for which $\Gamma$'s we have $q_c(\Gamma)\leq q$, for any fixed $q\in (0,1]$. By Theorem~\ref{th:Zadik}, it is sufficient to show that $ \Tilde{q}_E\leq  q/\log|e(\Gamma)|\triangleq \tilde{q}$ for any fixed $q$, which holds if and only if,
    \begin{align}
        \inf_{\s{H}\subseteq \Gamma} \frac{\E_{\s{G}\sim\calG(n,\tilde{q})}\pp{\calN(\s{H},\s{G})}}{\calN(\s{H},\Gamma)}\geq \frac{1}{2},\label{eq:condinf}
    \end{align}
    We note that,
    \begin{align}
    {\E_{\s{G}\sim\calG(n,\tilde{q})}\pp{\calN(\s{H},\s{G})}}=|\calS_{\s{H}}|\cdot \tilde{q}^{|\s{H}|},
    \end{align}
    where we recall that $|\calS_{\s{H}}|=\calN(\s{H},\calK_n)$ denotes the number of copies of $\s{H}$ in the complete graph, and $|\s{H}|$ denotes the number of edges in $\s{H}$. Thus, \eqref{eq:condinf} holds if and only if for any $\s{H}\subseteq \Gamma$ we have,
    \begin{align}
        \frac{\calN(\s{H},\Gamma)}{|\calS_{\s{H}}|}\leq 2 \tilde{q}^{|H|}=2\p{\frac{q}{\log|e(\Gamma)|}}^{|\s{H}|}.\label{eq:probRC}
    \end{align}
    An easy combinatorial argument shows that the expression on the left-hand-side of \eqref{eq:probRC} equals $\P_{\Gamma}[\s{H}\subseteq \Gamma]$ (see, \cite[Lemma~10]{elimelech2025detecting}). Thus, \eqref{eq:probRC} holds if,
    \begin{align}
        \log \P_{\Gamma}[\s{H}\subseteq \Gamma]-\log2-|\s{H}|\p{\log q -\log\log|e(\Gamma)|}\leq 0.
    \end{align}
    By the assumptions that $|v(\Gamma)|=\omega(1)$, and that $\Gamma$ has no isolated vertices, for any $\varepsilon$ and for sufficiently large $n$, the above holds if, 
    \begin{align}
        \log \P_{\Gamma}[\s{H}\subseteq \Gamma]+(1+\varepsilon)|\s{H}|\log\log|e(\Gamma)|\leq 0.\label{eq:almostF}
    \end{align}
    Finally, using Lemma~\ref{lem:probcalc}, the left-hand-side of \eqref{eq:almostF} can be upper bounded by,
    \begin{align}
        \log& \P_{\Gamma}[\s{H}\subseteq \Gamma]+(1+\varepsilon)|\s{H}|\log\log|e(\Gamma)|\\
        &\leq |v(\s{H})|\p{\log|v(\Gamma)|-\log n}+(1+\varepsilon)|\s{H}|\log\log|e(\Gamma)|\\
        &=|v(\s{H})|\cdot \p{\log|v(\Gamma)| +(1+\varepsilon)\frac{|\s{H}|}{|v(\s{H})|} \log \log |e(\Gamma)| -\log n}\\
        &\overset{(a)}{\leq}|v(\s{H})|\cdot \p{\log|v(\Gamma)|+(1+\varepsilon)\mu(\Gamma) \log \log |e(\Gamma)| -\log n }\leq 0,
    \end{align}
    where $(a)$ follows from the definition of the maximal subgraph density, and the last inequality follows from \eqref{eq:condCopy}. This concludes the proof.
\end{proof}
\begin{remark} As mentioned at the end of Section~\ref{sec:StatLim}, a gap remains for planted subgraphs with sub-logarithmic (but not strongly sub-logographic) density, where  \begin{align}
    O\p{\frac{\log|v(\Gamma)|}{\log \log|v(\Gamma)|}}\leq \mu(\Gamma)\leq o(\log|v(\Gamma)|).
\end{align}
This gap originates from the term $\log|e(\Gamma)|$ in the upper bound used for the critical probability \eqref{eq:critProb}. While it is tempting to believe that sharper bounds on the critical probability may be sufficient to close this gap, a close examination of the critical probability suggests that this may not be the case. Indeed, the upper bound in \eqref{eq:critProb}, proved in~\cite{mossel2022second} is a relaxed version of the bound conjectured by Kahn and Kalai \cite{kahn2007thresholds}, known to be tight for many families of graphs:
\begin{align}
    q_c(\Gamma)\leq C\cdot q_E\cdot  \log|e(\Gamma)|,
\end{align}
for some $C>0$, where 
\begin{align}
    q_E(\Gamma)\triangleq\min\ppp{q\in[0,1] ~\bigg|~ \E_{\s{G}\sim \calG(n,q)}\pp{\calN(\s{H},\s{G})}\geq \frac{1}{2}\;\s{for}\;\s{all}\;\s{H}\subseteq \Gamma}.
\end{align}
However, plugging the improved Kahn-Kalai (conjectured) bound into our analysis reveals that the $\p{\log\log|v(\Gamma)|}^{-1}$ factor remains. This suggests that perhaps an entirely different argument may be required to close this gap.
\end{remark}

\section{Proofs of Upper Bounds}\label{sec:UB}

\subsection{Proof of Theorem~\ref{th:advOpt}}

Let $\s{A}$, $\s{A}_{\s{adv}_0}$, and $\s{A}_{\s{adv}_1}$, denote the adjacency matrices of the observed graphs in the detection problem without adversary \eqref{eqn:super_hypo} and with adversary \eqref{eqn:decMain}, respectively. Furthermore, let $\calO(\s{A},\bar{\Gamma})\triangleq\sum_{(i,j)\in\bar{\Gamma}}\s{A}_{ij}$, and as so, $T_{\s{scan}}(\s{A})\triangleq \max_{\bar{\Gamma} \in {\calS_{\Gamma_{\max}}}}\calO(\s{A},\bar{\Gamma})$. We start with the analysis of the Type-I error probability. First, note that due to the monotonicity of both the adversary's action and the scan test objective function, we have,
\begin{align}
    T_{\s{scan}}(\s{A}_{\s{adv}_0}) &= \calO\p{\s{A}_{\s{adv}_0},\bar{\Gamma}^{\s{adv}}_{\s{max}}}\leq \calO\p{\s{A},\bar{\Gamma}^{\s{adv}}_{\s{max}}}\leq \calO\p{\s{A},\bar{\Gamma}_{\s{max}}} = T_{\s{scan}}(\s{A}),
\end{align}
where $\bar{\Gamma}^{\s{adv}}_{\s{max}}$ denotes the maximizer in \eqref{eqn:scanstat} when applied on $\s{A}_{\s{adv}_0}$, while $\bar{\Gamma}_{\s{max}}$ is the maximizer in \eqref{eqn:scanstat} when applied on $\s{A}$. Therefore, 
\begin{align}
    \pr_{\calH_0}\pp{\psi_{\s{scan}}(\s{A}_{\s{adv}_0})=1}&=\pr_{\calH_0}\pp{T_{\s{scan}}(\s{A}_{\s{adv}_0})\geq\tau_{\s{scan}}}\\
    &\leq \pr_{\calH_0^{\s{na}}}\pp{T_{\s{scan}}(\s{A})\geq\tau_{\s{scan}}}\to0,
\end{align}
where the decay of $\pr_{\calH_0^{\s{na}}}\p{T_{\s{scan}}(\s{A})\geq\tau_{\s{scan}}}$ follows from Theorem~\ref{thm:upperBoundAlgo}. 

Next, we analyze the Type-II error probability. Under $\calH_1$, since the adversary keeps the edges of the planted subgraph $\Gamma$ intact, there exists a subgraph $\bar{\Gamma}^\star\in\calS_{\Gamma_{\max}}$, such that $\calO(\s{A}_{\s{adv}_1},\bar{\Gamma}^\star)=\sum_{(i,j)\in\bar{\Gamma}^\star}[\s{A}_{\s{adv}_1}]_{ij} \sim \mathsf{Binomial}{\left(|e{\left(\bar{\Gamma}\right)|},p\right)}$, and, in particular, $\calO(\s{A}_{\s{adv}_1},\bar{\Gamma}^\star)=\calO(\s{A},\bar{\Gamma}^\star)$, in distribution. 
Therefore, by Chebyshev's inequality, for $\kappa<p$,
\begin{align}
\pr_{\calH_1}\pp{\psi_{\s{scan}}(\s{A}_{\s{adv}_1})=0} &= \pr_{\calH_1}\pp{T_{\s{scan}}(\s{A}_{\s{adv}_1})\leq\tau_{\s{scan}}}\\
& = \sum_{\Gamma_0\in\calS_{\Gamma}}\pr_{\calH_1}\pp{\left.\max_{\bar{\Gamma}\in\calS_{\Gamma_{\max}}}\calO(\s{A}_{\s{adv}_1},\bar{\Gamma})<\kappa\cdot |e(\Gamma_{\max})| \right| \Gamma=\Gamma_0}\cdot \pr\pp{\Gamma=\Gamma_0}\\
&\leq \sum_{\Gamma_0\in\calS_{\Gamma}}\pr_{\calH_1}\pp{\left.\calO(\s{A}_{\s{adv}_1},\bar{\Gamma}^\star(\Gamma_0))<\kappa\cdot |e(\Gamma_{\max})|\right| \Gamma=\Gamma_0}\cdot \pr\pp{\Gamma=\Gamma_0}\\
&= \sum_{\Gamma_0\in\calS_{\Gamma}}\pr_{\calH_1^{\s{na}}}\pp{\left.\calO(\s{A},\bar{\Gamma}^\star(\Gamma_0))<\kappa\cdot |e(\Gamma_{\max})|\right| \Gamma=\Gamma_0}\cdot \pr\pp{\Gamma=\Gamma_0}\\
&\leq \sum_{\Gamma_0\in\calS_{\Gamma}}\frac{|e(\Gamma_{\max})|p(1-p)}{(p-\kappa)^2|e(\Gamma_{\max})|^2}\cdot  \pr\pp{\Gamma=\Gamma_0}   \\
    & = \frac{p(1-p)}{(p-\kappa)^2|e(\Gamma_{\max})|}.
\end{align} 
Consider the convex combination $\kappa = \epsilon q+(1-\epsilon)p$. Note that for any $\epsilon\in(0,1]$ we have $q\leq\kappa<p$, as required above. For this choice,
\begin{align}
\pr_{\calH_1}\pp{\psi_{\s{scan}}(\s{A}_{\s{adv}_1})=0} &\leq \frac{p(1-p)}{\epsilon^2(p-q)^2|e(\Gamma_{\max})|}\to0,
\end{align} 
where the last passage follows under the condition $p\cdot |e(\Gamma_{\max})|\to\infty$, as stated in Theorem~\ref{thm:upperBoundAlgo}.\footnote{Note that the proofs of Theorems~\ref{thm:upperBoundAlgo} and \ref{th:advOpt} (see, in particular, \cite[eqns. (63)--(67)]{elimelech2025detecting}) hold more generally for any values of $p$ and $q$ that may depend on $n$, and not just for fixed constants.}

\subsection{Proof of Theorem~\ref{th:advCON0}}

\paragraph{Detection without adversary.} Consider first the vanilla testing problem of a general planted subgraph $\Gamma$ in \eqref{eqn:super_hypo}, and recall our test in \eqref{eqn:TestConObject}--\eqref{eqn:TestCon}. The following result states that $\psi_{\calC_{\Gamma'}}$ discriminants between $\calH_0^{\s{na}}$ and $\calH_1^{\s{na}}$, with high probability.
\begin{theorem}\label{thm:GennoAdv}
   Consider the $(\calH_0^{\s{na}},\calH_1^{\s{na}})$ detection problem in \eqref{eqn:super_hypo}. Then, $\s{R}(\psi_{\calC_{\Gamma'}})\to0$, provided that 
   \begin{align}
   \frac{|e(\Gamma')|}{\norm{\Gamma'}_\star}\geq C'_{p,q}\sqrt{n},    
   \end{align}
   where $C'_{p,q}\triangleq\frac{C}{p-q}$, for some universal $C>0$.
\end{theorem}
The proof of the above result relies on the following lemmas.
\begin{lemma}\label{lemma:1gen}
Fix $\delta>0$. Under $\calH_0^{\s{na}}$, with probability at least $1-\delta$,
\begin{align}
\calC_{\Gamma'}(\s{W})\leq \frac{C}{q}\norm{\Gamma'}_\star\sqrt{n}+\frac{C}{q}\norm{\Gamma'}_\star\sqrt{\log\frac{4}{\delta}},
\end{align}
for some $C>0$.
\end{lemma}
\begin{lemma}\label{lemma:2gen}
Fix $\delta>0$. Under $\calH_1^{\s{na}}$, with probability at least $1-\delta$,
\begin{align}
\calC_{\Gamma'}(\s{W})\geq 2\frac{p-q}{q}|e(\Gamma')|-\sqrt{\frac{|e(\Gamma')|}{2}\log\frac{1}{\delta}}.
\end{align}
\end{lemma}
Let us show first that Lemmas~\ref{lemma:1gen} and \ref{lemma:2gen} imply Theorem~\ref{thm:GennoAdv}.
\begin{proof}[Proof of Theorem~\ref{thm:GennoAdv}]
Using the upper bound in Lemma~\ref{lemma:1gen} and the lower bound in Lemma~\ref{lemma:2gen}, we get 
\begin{align}
\pr_{\calH_0^{\s{na}}}\pp{\calC_{\Gamma'}(\s{W})\geq\tau_0}\leq\tilde\delta,\ \ \ \ \pr_{\calH_1^{\s{na}}}\pp{\calC_{\Gamma'}(\s{W})\leq\tau_1}\leq\tilde\delta,\label{eqn:Ty12con}
\end{align}
where
\begin{align}
\tau_0&\triangleq \frac{C}{q}\norm{\Gamma'}_\star\sqrt{n}+\frac{C}{q}\norm{\Gamma'}_\star\sqrt{\log\frac{4}{\delta}},\\
\tau_1&\triangleq2\frac{p-q}{q}|e(\Gamma')|-\sqrt{\frac{|e(\Gamma')|}{2}\log\frac{1}{\delta}}.
\end{align}
Accordingly, whenever $\tau_1>\tau_0$ we take $\tau\in[\tau_0,\tau_1)$. Then, \eqref{eqn:Ty12con} implies that $\s{R}(\psi_{\calC_{\Gamma'}})\leq2\tilde\delta\triangleq\delta$. Finally, we note that $\tau_1>\tau_0$ holds when 
\begin{align}
    2\frac{p-q}{q}\frac{|e(\Gamma')|}{\norm{\Gamma'}_\star}\geq \frac{C}{q}\sqrt{n}+\sqrt{\frac{|e(\Gamma')|}{2\norm{\Gamma'}_\star^2}\log\frac{1}{\delta}}+\frac{C}{q}\sqrt{\log\frac{4}{\delta}}.\label{eqn:cconcc}
\end{align}
Now, since for any simple graph on $m$ edges with adjacency matrix $\s{A}$, we have $\norm{\s{A}}_\star\geq\norm{\s{A}}_F = \sqrt{2m}$ (see Appendix~\ref{app:LBNuc}), then, \eqref{eqn:cconcc} holds provided that,
\begin{align}
    2\frac{p-q}{q}\frac{|e(\Gamma')|}{\norm{\Gamma'}_\star}\geq \frac{C}{q}\sqrt{n}+\frac{C}{q}\sqrt{\log\frac{4}{\delta}},
\end{align}
which holds if $\frac{|e(\Gamma')|}{\norm{\Gamma'}_\star}\geq C'_{p,q,\delta}\sqrt{n}$, for some $C'_{p,q,\delta}>0$. Finally, taking $\delta = 1/n\to0$, we obtain that $\s{R}(\psi_{\calC_{\Gamma'}})\to0$ if $\frac{|e(\Gamma')|}{\norm{\Gamma'}_\star}>C'_{p,q}\sqrt{n}$, with $C'_{p,q} = \frac{C}{p-q}$.
\end{proof}
It is left to prove Lemmas~\ref{lemma:1gen} and \ref{lemma:2gen}.
\begin{proof}[Proof of Lemma~\ref{lemma:1gen}]
Under $\calH_0^{\s{na}}$ we wish to upper bound the SDP value. The idea is to drop the constraint $\mathbf{0}\leq\s{Z}\leq\mathbf{J}$ from \eqref{eqn:convex}, and show that the value of the resulting value of the optimization problem, which can only be bigger, is upper bounded by $\norm{\Gamma'}_\star\cdot\sigma_{\max}(\s{W})$, where $\sigma_{\max}(\s{W})$ is the largest singular-value of $\s{W}$. Let $\s{Z}$ be a solution to \eqref{eqn:convex}. Then
\begin{align}
\innerP{\s{W},\s{Z}}\leq\sup\ppp{\innerP{\s{W},\s{Y}}~:~ \norm{\s{Y}}_\star\leq\norm{\Gamma'}_\star, \s{Y}=\s{Y}^\top}\leq \norm{\Gamma'}_\star\cdot\sigma_{\max}(\s{W}).
\end{align}
Indeed, the first inequality holds because we have just relaxed SDP \eqref{eqn:convex}, while the second inequality holds by the following argument. Let $\s{Y} = \sum_{i=1}^n\mu_i\mathbf{u}_i\mathbf{v}_i^H$ be the singular value decomposition (SVD) of $\s{Y}$, where $\ppp{\mu_i}_i$ are the (non-negative) singular values of $\s{Y}$ and $\ppp{\mathbf{u}_i,\mathbf{v}_i}_i$ are the corresponding orthonormal left and right singular vectors. Note that in our case $\s{Y}$ is a symmetric matrix, and thus $\mathbf{u}_i = \mathbf{v}_i$ for $i\in[n]$; however, we retain the original notation for generality. Then, for any $\s{Y}$ with $\norm{\s{Y}}_\star\leq\norm{\Gamma'}_\star$,\footnote{{Equivalently, the bound in \eqref{eqn:BoundSpectralDual} follows from the duality $\sup_{\norm{\s{Y}}_{\star}\leq 1}\langle \s{W},\s{Y}\rangle=\norm{\s{W}}_{\s{op}}$:
\begin{align*}
\langle \s{W},\s{Y}\rangle
&\leq \sup_{\norm{\s{Y}'}_{\star} \leq \norm{\s{Y}}_{\star}}\langle \s{W},\s{Y}'\rangle
= \norm{\s{Y}}_{\star} \sup_{\norm{\s{Y}'}_{\star} \leq 1}\langle \s{W},\s{Y}'\rangle 
&= \norm{\s{Y}}_{\star} \norm{\s{W}}_{\s{op}}
\le \norm{\Gamma'}_{\star} \norm{\s{W}}_{\s{op}}
= \norm{\Gamma'}_{\star} \sigma_{\max}(\s{W}).
\end{align*}}}
\begin{align}
\innerP{\s{W},\s{Y}} &= \sum_{i=1}^n\mu_i\mathbf{v}_i^H\s{W}\mathbf{u}_i\leq\sigma_{\max}(\s{W})\cdot\sum_{i=1}^n\mu_i\leq \sigma_{\max}(\s{W})\cdot\norm{\Gamma'}_\star,\label{eqn:BoundSpectralDual}
\end{align}
where the first inequality follows from the definition of the spectral norm.
 Since $\s{W}$ is a symmetric i.i.d. matrix, with zero mean, it is well-known that (see, e.g., \cite[Corollary 4.4.7]{vershynin2010introduction}) with probability at least $1-\delta$,
\begin{align}
    \sigma_{\max}(\s{W})=|\lambda_{\max}(\s{W})|\leq \frac{C}{q}\sqrt{n}+\frac{C}{q}\sqrt{\log\frac{4}{\delta}},
\end{align}
for some universal $C>0$, which concludes the proof of the upper bound. 
\end{proof}

\begin{proof}[Proof of Lemma~\ref{lemma:2gen}]
Let $\s{Z}^\star$ denote the true underlying adjacency matrix of the planted version of $\Gamma'$ in $\calK_n$ (padded with zeroes). It is clear that $\s{Z}^\star$ is in the feasibility set of \eqref{eqn:convex}. Then, under $\calH_1^{\s{na}}$, we have
\begin{align}
\calC_{\Gamma'}(\s{W})\geq \innerP{\s{Z}^\star,\s{W}}.
\end{align}
Since $\innerP{\s{Z}^\star,\s{W}}\sim\frac{2}{q}\cdot\s{Binomial}(|e(\Gamma')|,p)-|e(\Gamma')|)$, from Hoeffding's inequality we have,
\begin{align}
    \innerP{\s{Z}^\star,\s{W}}\geq 2\frac{p-q}{q}|e(\Gamma')|-\sqrt{\frac{|e(\Gamma')|}{2}\log\frac{1}{\delta}}.\label{eqn:Hoeffgen}
\end{align}
with probability at least $1-\delta$.
\end{proof}

\paragraph{Detection with adversary.} We next show that \eqref{eqn:convex} is robust against the monotone adversary in our semi-random model. Let $\psi_{\calC_{\Gamma'}}(\s{W})\triangleq\Ind\ppp{\calC_{\Gamma'}(\s{W})>\tau}$, with $\tau\in[\tau_0,\tau_1)$, where 
\begin{align}
\tau_0\triangleq \frac{C}{q}\norm{\Gamma'}_\star\sqrt{n}+\frac{C}{q}\norm{\Gamma'}_\star\sqrt{\log\frac{4}{\delta}},\quad
\tau_1\triangleq 2\frac{p-q}{q}|e(\Gamma')|-\sqrt{\frac{|e(\Gamma')|}{2}\log\frac{1}{\delta}}.
\end{align}
We are now in a position to prove Theorem~\ref{th:advCON0}, which we restate here in a slightly more general form.
\begin{theorem}\label{th:advCON}
Consider the detection problem in \eqref{eqn:decMain}. Then, $\s{R}(\psi_{\calC_{\Gamma'}})\to0$, provided that 
\begin{align}
\frac{|e(\Gamma')|}{\norm{\Gamma'}_\star}\geq C'_{p,q}\sqrt{n},    
\end{align}
where $C'_{p,q}\triangleq\frac{C}{p-q}$, for some universal $C>0$.
\end{theorem}

\begin{proof}[Proof of Theorem~\ref{th:advCON}]
We start with the analysis of the Type-I error probability. First, note that due to the monotonicity of the adversary and the objective function in \eqref{eqn:convex}, we have,
\begin{align}
    \calC_{\Gamma'}(\s{W}_{\s{adv}}) &= \innerP{\s{W}_{\s{adv}},\hat{\s{Z}}^{\s{adv}}_{\calC_{\Gamma'}}}\leq \innerP{\s{W},\hat{\s{Z}}^{\s{adv}}_{\calC_{\Gamma'}}}\leq \innerP{\s{W},\hat{\s{Z}}_{\calC_{\Gamma'}}} = \calC_{\Gamma'}(\s{W}),
\end{align}
where $\hat{\s{Z}}^{\s{adv}}_{\calC_{\Gamma'}}$ denotes the maximizer of the SDP in \eqref{eqn:convex} when applied on $\s{W}_{\s{adv}}$, while $\hat{\s{Z}}_{\calC_{\Gamma'}}$ is the maximizer of the SDP in \eqref{eqn:convex} when applied on $\s{W}$, and the first inequality follows since $\hat{\s{Z}}_{\calC_{\Gamma'}}^{\s{adv}}$ has non-negative entries. Thus, 
\begin{align}
    \pr_{\calH_0}\pp{\calC_{\Gamma'}(\s{W}_{\s{adv}})\geq\tau}\leq \pr_{\calH_0^{\s{na}}}\pp{\calC_{\Gamma'}(\s{W})\geq\tau}\leq\delta,
\end{align}
where the second inequality follows from Theorem~\ref{thm:GennoAdv}. As for the Type-II error probability, recall that $\s{Z}^\star$ denotes the true underlying adjacency matrix of the planted subgraph $\Gamma^\star$. Then, we have,
\begin{align}
    \pr_{\calH_1}\pp{\calC_{\Gamma'}(\s{W}_{\s{adv}})\leq\tau}&\leq\pr_{\calH_1}\pp{\innerP{\s{Z}^\star,\s{W}}\leq\tau} =\pr_{\calH_1^{\s{na}}}\pp{\innerP{\s{Z}^\star,\s{W}}\leq\tau}\leq\delta,
\end{align}
where the second equality is because the adversary kept the planted subgraph $\Gamma^\star$ (or, equivalently, $\s{Z}^\star$) intact, and the last inequality follows for any $\tau<\tau_1$ (see, Lemma~\ref{lemma:2gen}, and in particular, \eqref{eqn:Hoeffgen}). Finally, by taking $\delta = 1/n\to0$ the proof is concluded.

\end{proof}

\section{Conclusion and Outlook}\label{sec:Out}

This work initiates the study of robust planted subgraph detection under a semi-random model, where an adversary can remove edges outside the planted subgraph. We establish sharp statistical thresholds for detectability and provide the first computationally efficient algorithms with provable robustness and statistical guarantees. Our results open a new direction in the study of inference under adversarial perturbations, bridging the gap between idealized models and real-world settings.

We hope our work opens more doors than it closes. Several interesting directions remain open for future research:
\begin{enumerate}
    \item \emph{Varying edge probabilities}: Our analysis focused on fixed edge probabilities $p$ and $q$. Extending results to the regime where these probabilities depend on $n$, especially when they decay to zero or one (e.g., polynomially), poses significant technical and conceptual challenges.
    \item \emph{Computational limits under adversaries}: It remains unclear whether the computational thresholds known for the classical (non-adversarial) setting can be achieved in the monotone adversary model. Addressing this requires new robust algorithms with provable guarantees—or the development of lower-bound frameworks tailored to semi-random models.
    \item \emph{Stronger adversaries}: {Our semi-random model considers a monotone adversary that can delete any edge not belonging to the planted subgraph $\Gamma$, while the planted edges are protected. Alternative models, for instance, allowing deletions anywhere subject to a budget constraint, allowing addition of a bounded number of edges, or perform more general perturbations, are all natural; analyzing such adversaries for general $\Gamma$ appears to require different techniques and we leave it for future work.}
    \item \emph{Beyond detection}: While we focused on the detection problem, robust recovery of arbitrary planted subgraphs remains an important and largely open problem.
    \item {\emph{Unknown parameters}: Our tests are stated under the common simplifying assumption that the edge probabilities $(p,q)$ and $\Gamma$ (or, at least, some of its characteristics) are known, which allows explicit threshold calibration. 
    \begin{itemize}
        \item \emph{Unknown $(p,q)$.} Under our semi-random model, the adversary may delete an arbitrary subset of edges outside $\Gamma$. As a result, the observed global edge density can be significantly \emph{smaller} than $q$ and is generally not an unbiased estimator of $q$. Without further restrictions on the adversary (e.g., a deletion budget, randomness, or other structural constraints), consistent estimation of $q$ from the observed graph need not be possible; in particular, the parameter may be non-identifiable from the observation alone. 
        Nonetheless, one is not limited to an ``estimate-and-plug-in'' approach: when $(p,q)$ are unknown, the problem can be formulated as testing composite hypotheses, and one may consider generalized likelihood ratio type procedures that optimize (or take a supremum) over the admissible null and alternative families (e.g., over the adversaries and over $(p,q)$ in a specified parameter set). Developing such adaptive procedures with provable guarantees under semi-random perturbations is an interesting direction for future work.
        \item \emph{Conservative calibration.} In applications one may still calibrate tests conservatively using coarse prior information on $(p,q)$. For example, if one has an a priori upper bound $q\leq q_{\max}$, then thresholds derived using $q_{\max}$ remain valid (though potentially conservative), since the adversary can only \emph{remove} background edges. More general adversary restrictions that could enable principled parameter estimation are an interesting direction for future work.
        \item \emph{Unknown/misspecified $\Gamma$.} Some procedures assume that the planted shape $\Gamma$ (or relevant summaries such as $\Gamma_{\max}$) is specified. Extending our results to composite alternatives in which $\Gamma$ is unknown (or only known to belong to a family of candidate motifs) is a natural direction: it typically introduces additional multiple-testing/complexity considerations and may change detection thresholds. Understanding robustness to mild misspecification of $\Gamma$ is also an interesting open problem.
    \end{itemize}}
    \item {\emph{Fine grained efficiency}: The SDP in Theorem~\ref{th:advCON0} is polynomial-time but operates over an $n\times n$ matrix variable, so generic interior-point solvers can be impractical for large $n$ (memory $\Theta(n^2)$ and super-cubic time in the worst case). In practice, one typically uses first-order methods or low-rank factorized approaches that exploit structure in SDP relaxations, and/or combines the method with preprocessing that reduces the candidate vertex set. A detailed empirical study and scalable implementations are left for future work.}
\end{enumerate}
We believe addressing these questions will further enhance our understanding of robust inference in semi-random models, and deepen the interplay between statistics, computation, and adversarial models.

\bibliographystyle{alpha}
\bibliography{bibfile}

\appendix

\section{Derivation of the Maximum Likelihood Estimator}\label{app:0}

Consider the following recovery task. Pick a copy $\Gamma^\star\in\calS_{\Gamma}$. A random graph $\s{G}$ with $n$ vertices is formed as follows: keep the edges of $\Gamma$ with probability $p$, and the edges outside $\Gamma$ with probability $q$. We denote the ensemble of such planted graphs by $\calG_{\Gamma^\star}(n,p,q)$. Given $\s{G}$, the goal is to recover the graph $\Gamma^\star$. Then, we have
\begin{align}
    \pr_{\calG_{\Gamma^\star}(n,p,q)}(\s{G};\Gamma) &= \prod_{(i,j)\in e(\Gamma)}p^{\s{A}_{ij}}(1-p)^{1-\s{A}_{ij}}\prod_{(i,j)\in \binom{[n]}{2}\setminus e(\Gamma)}q^{\s{A}_{ij}}(1-q)^{1-\s{A}_{ij}}\\
    & \propto \prod_{(i,j)\in e(\Gamma)}\pp{\frac{p(1-q)}{q(1-p)}}^{\s{A}_{ij}}\\
    & = \pp{\frac{p(1-q)}{q(1-p)}}^{\sum_{(i,j)\in e(\Gamma)}\s{A}_{ij}}
\end{align}
Thus, we see that in the regime $p>q$, the MLE is given by
\begin{align}
\hat{\Gamma}_{\s{MLE}} = \arg\max_{\Gamma\in\calS_\Gamma}\sum_{(i,j)\in e(\Gamma)}\s{A}_{ij}.\label{eqn:combiDev}
\end{align}

\section{Convex Relaxation of the Optimization Problem}\label{app:conRel}

In this appendix, we derive a convex relaxation of \eqref{eqn:OptiMax}. Let $\s{A}$ denote the adjacency matrix of $\Gamma$. Then, note that the maximization problem in \eqref{eqn:OptiMax} can be rewritten as,
\begin{align}
    \s{OPT}_{\s{d}}\triangleq\max_{\substack{\s{A'}\leq\s{A}\\ \s{A}'\in\{0,1\}^{n\times n},\;\s{A}'\neq\mathbf{0}}}\frac{\innerP{\s{A}',\s{A}}}{\norm{\s{A}'}_\star}.\label{eqn:discr}
\end{align}
Then, we define the relaxation
\begin{align}
    \s{OPT}_{\s{relax}}\triangleq\max_{\substack{\s{X}\in\mathbb{R}^{n\times n}\\ \norm{\s{X}}_\star\leq1,\; \mathbf{0}\leq\s{X}\leq \mathbf{J},\; \s{X}_{ij}=0\; \s{if}\;\s{A}_{ij}=0}}\innerP{\s{X},\s{A}}.\label{eqn:Condiscr}
\end{align}
Note that the relaxed problem above is convex: the objective is linear, the nuclear-norm constraint is convex, and the support constrains are linear. We prove the following result.
\begin{prop}\label{prop:dd}
    Fix $\Gamma\in\calS_\Gamma$ with adjacency matrix $\s{A}$. Let $\s{OPT}_{\s{d}}$ and $\s{OPT}_{\s{relax}}$ be defined in \eqref{eqn:discr} and \eqref{eqn:Condiscr}, respectively. Then, $\s{OPT}_{\s{d}}\leq\s{OPT}_{\s{relax}}$.
\end{prop}
\begin{proof}[Proof of Proposition~\ref{prop:dd}]
    We show that for any feasible matrix $\s{A}'$ for \eqref{eqn:discr} we can find a matrix $\s{X}$ feasible for \eqref{eqn:Condiscr}. Specifically, for any $\s{A}'$ we set $\s{X} = \frac{\s{A}'}{\norm{\s{A'}}_\star}$. Then, we note that $\norm{\s{X}}_\star\leq1$, and furthermore, $\innerP{\s{X},\s{A}} = \frac{\innerP{\s{A}',\s{A}}}{\norm{\s{A}'}_\star}$. Moreover, $\s{A}'\leq\s{A}$ implies that $\s{X}_{ij}=0$ whenever $\s{A}_{ij}=0$. Now, since $\s{A}'\in\{0,1\}^{n\times n}$ and $\s{A}'\neq\mathbf{0}$, we also have
    \begin{align}
        \norm{\s{A}'}_\star\geq\norm{\s{A}'}_F\geq\sqrt{2},
    \end{align}
    $\norm{\cdot}_F$ is the Frobenius norm. Therefore, for any $i,j\in[n]$, we have $0\leq \s{X}_{ij}\leq\frac{1}{\sqrt{2}}\leq1$, and as so the constraint $\mathbf{0}\leq\s{X}\leq \mathbf{J}$ holds. Thus, $\s{X}$ is feasible for \eqref{eqn:Condiscr} and satisfies
    \begin{align}
        \frac{\innerP{\s{A}',\s{A}}}{\norm{\s{A}'}_\star}=\innerP{\s{X},\s{A}}\leq\s{OPT}_{\s{relax}}.
    \end{align}
    Taking the maximum over all feasible $\s{A}'$ concludes the proof.
   
\end{proof}

\section{SDP Relaxation for Planted Clique}\label{app:1}

In this appendix, we pay special attention to the case where $\Gamma$ is a clique on $k$ vertices and provide a robust SDP-based detection test. The optimization procedure in \eqref{eqn:combi} can be expressed as follows,
\begin{equation}
\begin{aligned}
\hat{\s{Z}}_{\mathrm{MLE}}=& \underset{\s{Z}\in\{0,1\}^{n\times n}}{\arg\max}
& & \innerP{\s{W},\s{Z}} \\
& \ \text{s.t.}
& &  \s{Z}\succeq 0\\
&&&  \mathrm{rank}(\s{Z})=1\\
&&& \s{Z}_{ii}\leq 1,\;\forall i\in[n]\\
&&& \s{Z}_{ij}\geq 0,\;\forall i,j\in[n]\\
&&& \innerP{\Ib,\s{Z}}=k\\
&&& \innerP{\Jb,\s{Z}}=k^2,
\end{aligned}\label{eqn:combi2}
\end{equation}
where $\Ib$ and $\Jb$ are the identity and all-one matrices, respectively. Indeed, first, it is immediate that for all $\zeta$ in the feasible $\Xi$ of \eqref{eqn:combi}, the gram matrix $\zeta\zeta^\top$, is in the feasible set of \eqref{eqn:combi2}. Thus, it is left to show the other direction; we need to show that any $\s{Z}$ in the feasible of \eqref{eqn:combi2} can be written as $\zeta\zeta^\top$, for some $\zeta\in\Xi$. This follows from the following arguments. The constraints that $\s{Z}\succeq 0$ and $\mathrm{rank}(\s{Z})=1$ readily imply that $\s{Z} = \mathbf{x}\mathbf{x}^\top$, for some $\mathbf{x}\in\mathbb{R}^n$. Furthermore, the constraints $\innerP{\Ib,\s{Z}}=k$ and $\innerP{\Jb,\s{Z}}=k^2$ imply that $\sum_{i=1}^nx_i^2=k$ and $\abs{\sum_{i=1}^nx_i}=k$. The constraint $Z_{ii}\leq 1$ implies that $|x_i|\leq 1$, and the constraint that $Z_{ij}\geq0$ implies that \emph{all} $x_i$'s are either positive or negative. Assume that they are all non-negative. Then, when combining the above we get that $\sum_{i=1}^nx_i^2=k$ and $\sum_{i=1}^nx_i=k$, which coexist only when $x_i\in\{0,1\}$. If all \emph{all} $x_i$'s are negative, then we will reach to the same conclusion (up to a sign). Therefore, $\mathbf{x}\in\Xi$, which concludes the proof. It is a standard practice that by dropping the rank-one constraint leads to the following convex relaxation of \eqref{eqn:combi2}, 
\begin{equation}
\begin{aligned}
\hat{\s{Z}}_{\mathrm{SDP}}=& \underset{\s{Z}\in\mathbb{R}^{n\times n}}{\arg\max}
& & \innerP{\s{W},\s{Z}} \\
& \ \text{s.t.}
& &  \s{Z}\succeq 0\\
&&& \s{Z}_{ij}\geq 0,\;\forall i,j\in[n]\\ 
&&& \innerP{\Ib,\s{Z}}= k\\
&&& \innerP{\Jb,\s{Z}}=k^2.
\end{aligned}\label{eqn:SDP}
\end{equation}
The above is a semi-definite programming optimization problem that can be solved in polynomial-time. Let $\s{Z}^\star = \xi^\star(\xi^\star)^\top$ denote the adjacency matrix of $\Gamma$. Below, for simplicity of notations, we focus on the vanilla case where $(p,q)=(1,1/2)$, with the understanding that the generalization for any values of $(p,q)$ is straightforward.

\paragraph{SDP with no adversary.} We now propose and analyze a detection algorithm based on the SDP optimization problem in \eqref{eqn:SDP}, starting with the case where no adversaries exist; we let $(\calH_0^{\s{PC}},\calH_1^{\s{PC}})$ denote the standard planted clique hypothesis testing problem. Below, we let $\mathsf{SDP}(\s{W})$ denote the optimized objective function in \eqref{eqn:SDP}, i.e., $\mathsf{SDP}(\s{W})\triangleq\innerP{\s{W},\hat{\s{Z}}_{\mathrm{SDP}}}$. Then, define the SDP test as $\psi_{\s{SDP}}(\s{W})\triangleq\Ind\ppp{\mathsf{SDP}(\s{W})>\tau}$, for some threshold $\tau$ that we will determine later on. The following result states that this test discriminants between $\calH_0^{\s{PC}}$ and $\calH_1^{\s{PC}}$ with high probability.
\begin{theorem}\label{thm:PCnoAdv}
    Consider the $(\calH_0^{\s{PC}},\calH_1^{\s{PC}})$ detection problem. Then, $\s{R}(\psi_{\s{SDP}})\to0$, provided that $k>C'\sqrt{n}$, for some constant $C'>0$.
\end{theorem}
The proof of the above result relies on the following lemmas.
\begin{lemma}\label{lemma:1}
Fix $\delta>0$. Under $\calH_0^{\s{PC}}$, with probability at least $1-\delta$,
\begin{align}
\mathsf{SDP}(\s{W})\leq Ck\sqrt{n}+Ck\sqrt{\log\frac{4}{\delta}},
\end{align}
for some $C>0$.
\end{lemma}
\begin{lemma}\label{lemma:2}
Fix $\delta>0$. Under $\calH_1^{\s{PC}}$, with probability one,
\begin{align}
k^2\leq\mathsf{SDP}(\s{W}).
\end{align}
\end{lemma}
Let us show first that Lemmas~\ref{lemma:1} and \ref{lemma:2} imply Theorem~\ref{thm:PCnoAdv}.
\begin{proof}[Proof of Theorem~\ref{thm:PCnoAdv}]
Using the upper bound in Lemma~\ref{lemma:1} and the lower bound in Lemma~\ref{lemma:2}, we get 
\begin{align}
\pr_{\calH_0^{\s{PC}}}\pp{\mathsf{SDP}(\s{W})\geq\tau_0}\leq\tilde\delta,\ \ \ \ \pr_{\calH_1^{\s{PC}}}\pp{\mathsf{SDP}(\s{W})\leq\tau_1}\leq\tilde\delta,\label{eqn:Ty12}
\end{align}
where
\begin{align}
\tau_0\triangleq Ck\sqrt{n}+Ck\sqrt{\log\frac{4}{\tilde\delta}},\quad
\tau_1\triangleq k^2.
\end{align}
Accordingly, whenever $\tau_1>\tau_0$ we take $\tau\in[\tau_0,\tau_1)$. Then, \eqref{eqn:Ty12} implies that $\s{R}(\psi_{\s{SDP}})\leq2\tilde\delta\triangleq\delta$. Finally, we note that $\tau_1>\tau_0$ holds when $k>C_\delta\sqrt{n}$, for some $C_\delta>0$. Taking $\delta=1/n\to0$, with $C'=2C$, concludes the proof.
\end{proof}
It is left to prove Lemmas~\ref{lemma:1} and \ref{lemma:2}.
\begin{proof}[Proof of Lemma~\ref{lemma:1}]
Under $\calH_0^{\s{PC}}$ we wish to upper bound the SDP value. The idea is to drop the constraint $\norm{\s{Z}}_1=k^2$ from \eqref{eqn:SDP}, and show that the value of the resulting SDP, which can only be bigger, is actually $k\cdot\lambda_{\max}(\s{W})$. Indeed, let $\s{Z}$ be a solution to \eqref{eqn:SDP}, then,
\begin{align}
\innerP{\s{W},\s{Z}}\leq\sup\ppp{\innerP{\s{W},\s{Y}}: \s{Y}\succeq 0,\;\innerP{\Ib,\s{Y}}\leq k}\leq k\cdot\lambda_{\max}(\s{W}).
\end{align}
Indeed, the first inequality holds because we have just relaxed SDP \eqref{eqn:SDP}, while the second inequality holds by the following argument. Decompose $\s{Y} = \sum_{i=1}^n\mu_i\yb_i\yb_i^\top$, where $\ppp{\mu_i}_i$ are the (non-negative) eigenvalues of $\s{Y}$ and $\ppp{\yb_i}_i$ are the corresponding normalized eigenvectors. Then,
\begin{align}
\innerP{\s{W},\s{Y}} = \sum_{i=1}^n\mu_i\yb_i^\top\s{W}\yb_i\leq\lambda_{\max}(\s{W})\cdot\sum_{i=1}^n\mu_i\leq k\cdot\lambda_{\max}(\s{W}).
\end{align}
Next, since $\s{W}$ is a symmetric i.i.d. matrix, with $\s{Radamacher}(1/2)$ entries and zero mean, it is well-known that (see, e.g., \cite[Corollary 4.4.7]{vershynin2010introduction}) with probability at least $1-\delta$,
\begin{align}
    \lambda_{\max}(\s{W})\leq C\sqrt{n}+C\sqrt{\log\frac{4}{\delta}}
\end{align}
for some universal $C>0$, which concludes the proof of the upper bound. 
\end{proof}

\begin{proof}[Proof of Lemma~\ref{lemma:2}]
Let $\s{Z}^\star$ denote the true underlying adjacency matrix of the planted subgraph $\Gamma^\star$. It is clear that $\s{Z}^\star$ is in the feasibility set of \eqref{eqn:SDP}. Then, under $\calH_1^{\s{PC}}$, we have
\begin{align}
\mathsf{SDP}(\s{W})\geq \innerP{\s{W},\s{Z}^\star} = \sum_{i,j\in v(\Gamma)}\s{W}_{ij} = k^2,
\end{align}
with probability one.
\end{proof}

\paragraph{SDP with adversary.} We next show that \eqref{eqn:SDP} is robust against the monotone adversary in our semi-random model. Let $\psi_{\s{SDP}}(\s{W})\triangleq\Ind\ppp{\mathsf{SDP}(\s{W})>\tau}$, with $\tau\in[\tau_0,\tau_1)$, where 
\begin{align}
\tau_0\triangleq Ck\sqrt{n}+k\sqrt{2\log\frac{1}{\tilde\delta}},\quad
\tau_1\triangleq k^2.
\end{align}
We have the following result.
\begin{theorem}\label{th:adv}
Consider the detection problem in \eqref{eqn:decMain} where $\Gamma$ is a clique on $k$ vertices, and fix $\delta\in(0,1)$. Then, $\s{R}(\psi_{\s{SDP}})\leq\delta$, provided that $k>C_\delta\sqrt{n}$, for some constant $C_\delta$.
\end{theorem}

\begin{proof}[Proof of Theorem~\ref{th:adv}]
We start with the analysis of the Type-I error probability. First, note that due to the monotonicity of the adversary and the SDP objective function, we have,
\begin{align}
    \s{SDP}(\s{W}_{\s{adv}}) &= \innerP{\s{W}_{\s{adv}},\hat{\s{Z}}^{\s{adv}}_{\mathrm{SDP}}}\leq \innerP{\s{W},\hat{\s{Z}}^{\s{adv}}_{\mathrm{SDP}}}\leq \innerP{\s{W},\hat{\s{Z}}_{\mathrm{SDP}}} = \s{SDP}(\s{W}),
\end{align}
where $\hat{\s{Z}}^{\s{adv}}_{\mathrm{SDP}}$ denotes the maximizer of the SDP in \eqref{eqn:SDP} when applied on $\s{W}_{\s{adv}}$, while $\hat{\s{Z}}_{\mathrm{SDP}}$ is the maximizer of the SDP in \eqref{eqn:SDP} when applied on $\s{W}$. Thus, 
\begin{align}
    \pr_{\calH_0}\pp{\s{SDP}(\s{W}_{\s{adv}})\geq\tau}\leq \pr_{\calH_0^{\s{PC}}}\pp{\s{SDP}(\s{W})\geq\tau}\leq\delta,
\end{align}
where the second inequality follows from Theorem~\ref{thm:PCnoAdv}. As for the Type-II error probability, 
let $\s{Z}^\star$ denote the true underlying adjacency matrix of the planted subgraph $\Gamma^\star$. Then, we have,
\begin{align}
    \pr_{\calH_1}\pp{\s{SDP}(\s{W}_{\s{adv}})\leq\tau}&\leq\pr_{\calH_1}\pp{ \innerP{\s{W},\s{Z}^\star}\leq\tau} \\
    &=\pr_{\calH_1^{\s{PC}}}\pp{\innerP{\s{W},\s{Z}^\star}\leq\tau}\\
    & = \pr_{\calH_1^{\s{PC}}}\pp{\sum_{i,j\in v(\Gamma)}\s{W}_{ij}\leq\tau}=0,
\end{align}
where the second equality is because the adversary kept the planted clique intact, and the last inequality follows for any $\tau< k^2$ (see, Lemma~\ref{lemma:2}).   

\end{proof}

\section{Lower Bounds on the Nuclear Norm}\label{app:LBNuc}
In this section, we prove two lower bound on the nuclear norm of the adjacency matrix. Specifically, let $G = (V,E)$ be a simple undirected graph with $n$ vertices and $m$ edges. Let $\s{A}\in \{0,1\}^{n \times n}$ be its adjacency matrix. Since $\s{A}$ is symmetric, its eigenvalues $\lambda_1,\ldots,\lambda_n \in \mathbb{R}$, and its singular values are $|\lambda_1|, \ldots, |\lambda_n|$. The nuclear norm of $\s{A}$ is defined as
\begin{align}
\norm{\s{A}}_\star = \sum_{i=1}^n |\lambda_i|.
\end{align}
We derive three lower bounds on $\norm{\s{A}}_\star$. Note that,
\begin{align}
\norm{\s{A}}_F = \sqrt{\sum_{i,j\in[n]} \s{A}_{ij}^2} = \sqrt{2m}.
\end{align}
Thus, we get $\norm{\s{A}}_\star\geq\norm{\s{A}}_F= \sqrt{2m}$. Next, we prove another lower bound. Recall that the spectral norm is the largest singular value,
\begin{align}
\norm{\s{A}} = \max_{i\in[n]} |\lambda_i|,
\end{align}
and the nuclear norm satisfies,
\begin{align}
\norm{\s{A}}_\star\geq\frac{\sum_{i=1}^n \lambda_i^2}{\max_{i\in[n]} |\lambda_i|} = \frac{\norm{\s{A}}_F^2}{\norm{\s{A}}} = \frac{2m}{\norm{\s{A}}}.
\end{align}

\section{Worst-Case Adversarial Distributions}\label{app:UniformC}

Recall the setting and definitions given in the proof of Proposition~\ref{prop:boundWithProb}. 
\paragraph{Claim:} $d_{\s{TV}}\p{\pr_{\calH_0^\star\vert\cE},\pr_{\calH_1^\star}}=0$. 

\begin{proof}[Proof \#1]
    Let $\s{Cop}(\Gamma,\s{G})$ denote the (random) set of copies of $\Gamma$ in $\s{G}$ (which may be empty), and let $\Gamma_0^\star$ be the uniform random copy of $\Gamma$ in $\s{G}_{\s{Adv_0^\star}}$ chosen by the adversary in case that $\s{G}$ contains a copy of $\Gamma$. We claim that with respect to the measure $\P_{\calH_0\vert \cE}$ we have that $\Gamma_0^\star\sim \s{Unif}\p{\calS_\Gamma}$. Indeed, we observe that by the symmetry of the measure $\P_{\calH_0}$ (w.r.t. vertices and edges), all copies of $\Gamma$ in  $\calK_n$ are equally probable to be in  $\s{Cop}(\Gamma,\s{G})$. Thus, conditioned on $\cE$, for any copies $\Gamma_1,\Gamma_2\in \calS_\Gamma$ we have
     \begin{align}
         \P_{\calH_0\vert \cE}\pp{\Gamma_0^\star=\Gamma_1}&=\sum_{i=1}^{|\calS_\Gamma|}\P_{\calH_0}\pp{\left.\Gamma_0^\star=\Gamma_1, |\s{Cop}(\Gamma,\s{G})|=i\right|\cE}\\
         &=\sum_{i=1}^{|\calS_\Gamma|}\P_{\calH_0}\pp{\left.\Gamma_0^\star=\Gamma_1, |\s{Cop}(\Gamma,\s{G})|=i, \Gamma_1\in \s{Cop}(\Gamma,\s{G})\right|\cE}\\
         &=\sum_{i=1}^{|\calS_\Gamma|}\P_{\calH_0}\pp{\left.\Gamma_0^\star=\Gamma_1\right|\cE ,|\s{Cop}(\Gamma,\s{G})|=i, \Gamma_1\in \s{Cop}(\Gamma,\s{G})}\\
         &\hspace{3cm} \cdot \P_{\calH_0}\pp{ \left.|\s{Cop}(\Gamma,\s{G})|=i, \Gamma_1\in \s{Cop}(\Gamma,\s{G}) \right|\cE}\\
         &=\sum_{i=1}^{|\calS_\Gamma|} \frac{1}{i}\cdot \P_{\calH_0}\pp{ \left.|\s{Cop}(\Gamma,\s{G})|=i, \Gamma_1\in \s{Cop}(\Gamma,\s{G})\right|\cE}\\
         &=\sum_{i=1}^{|\calS_\Gamma|} \frac{1}{i}\cdot \P_{\calH_0}\pp{\left. |\s{Cop}(\Gamma,\s{G})|=i, \Gamma_2\in \s{Cop}(\Gamma,\s{G})\right|\cE}\\
         &=\P_{\calH_0\vert \cE}\pp{\Gamma_0^\star=\Gamma_2},
     \end{align}
     where the last equality follows by repeating the first steps in reverse (replacing $\Gamma_1$ by $\Gamma_2$), and the inequality before that follows from the symmetry argument. Thus, both the distribution of  $\s{G}_{\s{Adv}_0^\star}$ conditioned on $\calE$ and the distribution of  $\s{G}_{\s{Adv}_1^\star}$ are essentially generated by planting a copy of $\Gamma$ uniformly at random on an empty graph with $n$ vertices, which implies that indeed $d_{\s{TV}}\p{\pr_{\calH_0^\star\vert\cE},\pr_{\calH_1^\star}}=0$.
\end{proof}

\begin{proof}[Proof \#2]
    The claim above can also be viewed as a consequence of the characterization of permutation invariant measures on graphs. Indeed, note that under $\calG(n,q)$, conditioned on $\calE$, the copy $\Gamma_0^\star$ chosen by $\s{Adv}_0^\star$ is uniform over $\calS_\Gamma$. Indeed, for any permutation $\sigma\in\mathbb{S}_n$, let $\calT_\sigma$ relabel the vertices of a graph by $\sigma$. The law of $\calG(n,q)$ is $\mathbb{S}_n$-invariant, and the property $\calE$ is an isomorphism-invariant. Therefore, the conditional law $\calG(n,q)\vert\calE$ is also $\mathbb{S}_n$-invariant. Now, the rule of $\s{Adv}_0^\star$ is also invariant: applying $\sigma$ to the graph and then applying $\s{Adv}_0^\star$ is the same in the reversed order. Thus, the distribution of $\Gamma_0^\star$ is $\mathbb{S}_n$-invariant on $\calS_\Gamma$. Since the action of $\mathbb{S}_n$ on $\calS_\Gamma$ is transitive, the only invariant probability measure is the uniform one. 

To prove that $d_{\s{TV}}\p{\pr_{\calH_0^\star\vert\cE},\pr_{\calH_1^\star}}=0$ we use the following straightforward coupling idea. Let $\Pi$ be the following generative process on a graph with $n$ vertices: first, sample a copy $\Phi\sim\s{Unif}(\calS_\Gamma)$. Then, for each edge $e\in e(\Phi)$, include it independently with probability $p$, and remove all other edges not in $e(\Phi)$. The main observation here is that both $\s{Adv}_0^\star$ and $\s{Adv}_1^\star$ generate graphs distributed as $\Pi$:
\begin{enumerate}
    \item Under the alternative distribution with adversary $\s{Adv}_1^\star$, the graph is generated by planting a fixed copy of $\Gamma$ and keeping each of its edges independently with probability $p$, while deleting all other edges. Thus, $\pr_{\calH_1^\star}$ is equivalent to $\Pi$.
    \item Under $\calG(n,q)$ conditioned on $\calE$, and after applying $\s{Adv}_0^\star$, we already proved that $\Phi\sim\s{Unif}(\calS_\Gamma)$. Then, the adversary retains each edge in $\Phi$ independently with probability $p$, and removes all other edges. Thus, $\pr_{\calH_0^\star\vert\calE}$ is equivalent to $\Pi$ as well.
\end{enumerate}
Since both adversary distributions are generated by the same process $\Pi$, it follows that their distributions are equal, and so, $d_{\s{TV}}\p{\pr_{\calH_0^\star\vert\cE},\pr_{\calH_1^\star}}=0$.
\end{proof}
\end{document}